%% file: Main-arXiv.tex
\documentclass[9pt]{article}
\usepackage[a4paper, top=35mm, bottom=35mm, left=30mm, right=30mm]{geometry}

\usepackage{amsmath, amssymb, bm}
\usepackage{graphicx}
\usepackage[utf8]{inputenc} 
\usepackage[T1]{fontenc} 
\usepackage{natbib}
\usepackage{xcolor}
\usepackage{url}
\usepackage{lineno}
\usepackage[colorlinks=false, pdfborder={0 0 0}]{hyperref}
\usepackage[shortlabels]{enumitem}

\usepackage{float}

\newtheorem{theorem}{Theorem}
\newtheorem{lemma}{Lemma}
\newtheorem{proposition}{Proposition}
\newtheorem{corollary}{Corollary}
\newtheorem{condition}{Condition}

\newenvironment{proof}[1][]
  { \par\noindent\textbf{Proof\if\relax\detokenize{#1}\relax\else\ #1\fi.}\quad }
  { \hfill$\square$\par }

\newcommand{\qed}{\hfill \ensuremath{\square}}

\newenvironment{keywords}{
  \vspace{1em}
  \noindent\textbf{Keywords:}%
}{\par}

\newtheorem{remark}{Remark}

\input{Notations} 

\title{\bfseries Adaptive tail index estimation: minimal assumptions and non-asymptotic guarantees}

\author{
  Johannes Lederer\thanks{Universität Hamburg, Germany. Email: \texttt{johannes.lederer@uni-hamburg.de}} \and
  Anne Sabourin\thanks{Université Paris Cité, CNRS, MAP5, Paris; Télécom Paris, Institut polytechnique de Paris. Email: \texttt{anne.sabourin@math.cnrs.fr}} \and
  Mahsa Taheri\thanks{Universität Hamburg, Germany. Email: \texttt{mahsa.taheri@uni-hamburg.de}}
}

\date{}

\begin{document}
\maketitle

\newcommand{\jl}[1]{\textcolor{red}{#1}}

\begin{abstract}
  A notoriously difficult challenge in extreme value theory is the
  choice of the number $k\ll n$, where $n$ is the total sample size,
  of extreme data points to consider for inference of tail
  quantities. 
  Existing theoretical guarantees for adaptive methods typically
  require second-order assumptions or von Mises assumptions that are
  difficult to verify and often come with tuning parameters that are
  challenging to calibrate.  This paper revisits the problem of
  adaptive selection of $k$ for the Hill estimator. Our goal is not an
  `optimal'~$k$ but one that is `good enough', in the sense that we
  strive for non-asymptotic guarantees that might be sub-optimal but
  are explicit and require minimal conditions.  We propose a
  transparent adaptive rule that does not require preliminary
  calibration of constants, inspired by `adaptive validation'
  developed in high-dimensional statistics.  A key feature of our
  approach is the consideration of a grid for \( k \) of size
  \( \ll n \), which aligns with common practice among practitioners
  but has remained unexplored in theoretical analysis.  Our rule only
  involves an explicit expression of a variance-type term; in
  particular, it does not require controlling or estimating a bias
  term. 
  Our theoretical analysis is valid for all heavy-tailed
  distributions, specifically for all regularly varying survival
  functions. Furthermore, when von Mises conditions hold, our method
  achieves `almost' minimax optimality with a rate of
  \( \sqrt{\log \log n}~ n^{-|\rho|/(1+2|\rho|)}\) when the grid size
  is of order $\log n$, in contrast to the
  \( (\log \log (n)/n)^{|\rho|/(1+2|\rho|)} \) rate in existing
  work. Our simulations show that our approach performs particularly
  well for ill-behaved distributions.

\end{abstract}


\section{Introduction}

\input{Contents2/Introduction}

\input{Contents2/SectionAV}

\input{Contents2/SectionHill}

\input{Contents2/SecondOrder}

\section{Numerical experiments}\label{sec:sim}
\input{Contents2/Simulations}

\section{Discussion}

In this work, we propose a practical and computationally efficient method for selecting the tuning parameter $k$ of the Hill estimator in EVT. Rather than exploring all possible values of $k$, we limit our search to a smaller logarithmic grid, a common practice in real-world applications. This reduces computational complexity significantly.
Theoretically, our AV method provides statistical guarantees with minimal assumptions—eliminating the need for second-order conditions or von Mises-type assumptions. To the best of our knowledge, this represents the first such result in EVT. Moreover, when von Mises conditions are met, our approach achieves nearly minimax optimality.
Our simulations reveal that our approach outperforms~\cite{boucheron2015tail} and~\cite{drees1998selecting} when dealing with ill-behaved distributions.
Our framework is also flexible and could be extended beyond the Hill estimator to other EVT problems where error bounds take the form $\text{error} \leq V+B$. Unlike existing adaptive Hill-based methods, the AV method relies on weaker assumptions, paving the way for future research. We plan to explore its applicability to other extreme value estimators, leveraging recent advances in the field.

\section*{Acknowledgement}
J.L. and M.T. are grateful for partial funding by the Deutsche Forschungsgemeinschaft (DFG, German Research Foundation) under project numbers 502906238,  543964668, and 520388526 (TRR391). A.S. acknowledges the funding of the French ANR grant EXSTA, ANR-23-CE40-0009-01.  The authors thank Stéphane Boucheron, John Einmahl, and Holger Drees for constructive discussions at different stages of the project.

\section*{Supplementary material}
\label{SM}

\appendix

\section{Appendix~1 (Proofs)}
In this section, we provide proofs for some of the results stated in the main text.
\input{Contents2/Proofs}
\section*{Appendix~2 (Complementary results)}
\input{Contents2/Appendix}

\input{Contents/Addsim}

\bibliographystyle{biometrika}
\bibliography{bibliK}

\end{document}

%% file: Notations.tex
\numberwithin{equation}{section}

\usepackage{xspace}

\newcommand\ie{\emph{i.e.}\xspace}

\newcommand{\nset}{\mathbb{N}}
\newcommand{\rset}{\mathbb{R}}

\newcommand{\ind}{\mathbf{1}}
\newcommand{\un}{\ind}
\newcommand{\point}{\,\cdot\,}
\newcommand{\eqd}{\overset{d}{=}}
\DeclareMathOperator*{\argmin}{arg\,min}
\usepackage{xifthen}
\newcommand{\PP}[1][]{\ifthenelse{\equal{#1}{}}{\ensuremath{\mathbb{P}}}{\ensuremath{\mathbb{P}\left( #1 \right)}}}
\newcommand{\EE}[1][]{\ifthenelse{\equal{#1}{}}{\ensuremath{\mathbb E}}{\ensuremath{{\mathbb E}\left[ #1 \right]}}}
\newcommand{\ud}{\,\mathrm{d}} 

 \newcommand{\draftcolor}{black}
\newcommand{\FrechetPar}{\gamma}

\newcommand{\gammaS}{\gamma}

\newcommand{\kAV}{\textcolor{\draftcolor}{\hat{k}_{\operatorname{EAV}}}}

\newcommand{\HillEst}{\textcolor{\draftcolor}{\ensuremath{\hat{\gamma}}}}
\newcommand{\indexk}{\textcolor{\draftcolor}{\ensuremath{k}}}

\newcommand{\MSE}{{\operatorname{MSE}}}

\newcommand{\kbch}{\hat{k}_{\operatorname{BT}}}
\newcommand{\kdk}{\hat{k}_{\operatorname{DK}}}
\newcommand{\kS}{\textcolor{\draftcolor}{\ensuremath{k^*}}}

\newcommand{\grid}{\mathcal{K}}
\newcommand{\deltaG}{\delta_{\grid}}

\newcommand{\datain}{\textcolor{\draftcolor}{\ensuremath{X}}}

\newcommand{\datainiO}{\textcolor{\draftcolor}{\ensuremath{\datain_{(i)}}}}

\newcommand{\nbrsamp}{\textcolor{\draftcolor}{\ensuremath{n}}}
\newcommand{\gamaOrc}{\textcolor{\draftcolor}{\ensuremath{\gamma}}}
\newcommand{\quantile}{\textcolor{\draftcolor}{\ensuremath{\delta}}}
\newcommand{\variance}{\textcolor{\draftcolor}{\ensuremath{V}}}
\newcommand{\bias}{\textcolor{\draftcolor}{\ensuremath{B}}}
\newcommand{\varianceF}[1]{\textcolor{\draftcolor}{\variance(#1)}}
\newcommand{\biasF}[1]{\textcolor{\draftcolor}{\bias(#1)}}

\newcommand{\ZK}{\textcolor{\draftcolor}{Z_k}}

\newcommand{\Zk}{\textcolor{\draftcolor}{\ensuremath{Z_k}}}

\makeatletter 

\newcommand*{\deqN}{\mathrel{\rlap{%
			\raisebox{0.3ex}{$\m@th\cdot$}}%
		\raisebox{-0.3ex}{$\m@th\cdot$}}=}
	\newcommand*{\eqdN}{=\mathrel{\rlap{%
				\raisebox{0.3ex}{$\m@th\cdot$}}%
			\raisebox{-0.3ex}{$\m@th\cdot$}}}
	
	\makeatletter 

%% file: Contents2/Introduction.tex
Extreme value theory (EVT) provides a statistical framework for the
analysis of tail events.  The results of EVT are essential in
applications ranging from finance as far as environmental
sciences. Statistical methods in EVT typically rely on choosing a
small fraction of the available data to make inference about tail
quantities. In univariate peaks-over-threshold methods, given a
dataset $(X_i, i\le n)$ of independent observations with a common
cumulative distribution function $F$, this amounts to choosing the
number $k$ of the largest order statistics
$(X_{(1)}\ge X_{(2)} \ge \dots X_{(n)})$ retained for inference. Here
and throughout, we order the samples by decreasing order of magnitude
and the number $k$ is referred to as the `extreme sample size', as it
represents  the number of observations used in practice for
inference.

The minimal working assumption in EVT is that for some sequences
$a_n>0, b_n \in\rset$, for $x\in\rset$, $F^n(a_n x+b_n)\to G(x)$ as
$n\to \infty$, where $G$ is an extreme value distribution
$G$ of the type $G(x) = \exp[-(1+ \gamma x)_+^{-1/\gamma}]$. 
The distribution $F$ is then said to  belong to the max-domain of attraction of $G$ and $\gamma\in\rset$ is called the tail index.
The Fréchet case $\gamma>0$ corresponds to  heavy-tailed distributions, for which the survival function $\bar F= 1-F $ is regularly varying, that is
\begin{equation}
  \label{eq:rv-definition}
  \bar F (tx) /\bar F (t) \xrightarrow[]{} x^{-1/\gamma} \text{ as } t\to\infty,
\end{equation}
for any fixed $x>0$.  Standard reference textbooks on EVT and its applications include~\cite{embrechts2013modelling,beirlant2006statistics,Resnick2007,resnick2008extreme}. Estimating the tail index in this context has been the subject of a wealth of works, one main reason being that estimating $\gamma$ also allows to estimate high quantiles with the help of the Weissman estimator~\citep{weissman1978estimation}. A particularly popular estimator for $\gamma$ if the Hill estimator proposed in \cite{Hill1975} 
\begin{equation}
  \label{eq:hill}
  \HillEst(\indexk) ~=~ \frac{1}{\indexk}\sum_{i=1}^k \log\bigl( \datainiO / \datain_{(\indexk + 1)} \bigr)\,.
\end{equation}
Tail-index estimation has  been the focus of an extensive body of research over the past decades. Providing a comprehensive review of the literature is out of our scope and we refer the reader to the review paper~\cite{fedotenkov2020review}.  
A major bottleneck in  applications of EVT is that one must choose the threshold (or $k$) to hopefully reach a bias-variance compromise. Too small a $k$ leads to a high variance, whereas too large a $k$ may result in bias.  This `choice of $k$' issue is  a pervasive challenge in EVT, of course not limited to the Hill estimator. 
Regarding tail index estimation, \cite{scarrott2012review} provide an overview of proven methods for threshold selection, from rules of thumb and  graphical diagnostics~\citep{BVT1996} to  resampling based approaches~\citep{gomes2001bootstrap}, see also \cite{gomes2012adaptive,gomes2016bootstrap}, and automatic adaptive procedures based on second order assumptions with asymptotic guarantees \citep{hall1985adaptive,drees1998selecting,drees2001minimax}. As emphasized in \cite{scarrott2012review}, a significant limitation of the latter approaches is the second order assumptions and the associated parameter which may be difficult to estimate.  \cite{grama2008statistics} constitutes a notable exception insofar as the working assumption is an `accompanying Pareto tail' which does not imply nor is implied by~\eqref{eq:rv-definition}, although the rates of convergence obtained on
concrete examples indeed require second order or von Mises assumptions.

All the results mentioned thus far are of asymptotic nature.
Non-asymptotic  guarantees in EVT are much scarcer than asymptotic ones, although this emerging topic has been developing fast over the past decade, see \emph{e.g.}~\cite{carpentier2015adaptive,boucheron2015tail,goix2015learning,lhaut2021uniform,engelke2021learning,drees2019principal,clemenccon2021concentration,aghbalou2024cross} and the references therein. Beyond the principled benefit of being valid for finite (observable) sample sizes, the techniques of proofs involved in these analyses typically  require less regularity conditions than asymptotic studies. Indeed  it is relatively easy in this framework  to separate bias terms stemming from the non-asymptotic nature of the data at hand, from deviation terms (which we refer to as  `variance' terms by a slight abuse of terminology) arising from the natural variability of the considered process, above a fixed threshold. The variance  term may then be analysed independently  from the bias term, and most importantly, it is not technically required that the bias be negligible compared with the variance, mainly because Slutsky-type arguments are not needed.  
As an example, existence of  partial derivatives of the standard tail-dependence function (a functional measure of dependence of extremes) is not required in \cite{goix2015learning}, and second-order assumptions are unnecessary in~\cite{clemenccon2021concentration}.

The primary goal of this work is to leverage the advantages of a non-asymptotic analysis to develop an adaptive selection rule for \( k \). Our method provides certain 'optimality' guarantees under the minimal Condition~\eqref{eq:rv-definition}. In addition, it achieves strong (nearly  minimax, up to a power of $\log\log n$) guarantees under the more restrictive von Mises Condition~\ref{cond:vonMises}. 
A key feature of our approach is the consideration of a grid $\grid$ for \( k \) of size \(|\grid| \ll n \), which aligns with common practice among practitioners but has remained unexplored in theoretical analysis. 
The core strength of this work lies in its robust, data-driven selection of $k$, enabled by expressing the error as an exact quantile of a known distribution, with the remainder absorbed into a bias term that does not need to be specified explicitly.
To the best of our knowledge, no existing statistical method relying on EVT has achieved this.  
For simplicity, we focus on the flagship example of the Hill estimator in this paper.  
Our method does not replace classical asymptotic approaches in EVT, but rather offers a robust, assumption-light complement—an especially valuable addition given that strong parametric assumptions are often seen as a key limitation in practical applications of extreme value analysis.

The only existing pieces of work regarding adaptive tail index estimation in a
non-asymptotic framework are \cite{carpentier2015adaptive}
and~\cite{boucheron2015tail}. The latter authors derive tail bounds
for the Hill estimator, while the former propose a novel estimator
which may be seen as a generalization of Pickands estimator
\citep{pickands1975statistical}. Both
references 
propose an adaptive selection rule for $k$ with minimax guarantees.
Following in the footsteps of previous asymptotic studies, their
theoretical analysis requires stronger assumptions
than~\eqref{eq:rv-definition}. Namely \cite{carpentier2015adaptive}
work under a second order Pareto assumption,
$|1 -F(x) - C x^{-1/\gamma} | \le C' x^{-(1+\beta)/\gamma}$ where
$\beta$ is the second order parameter. 
On the other hand, \cite{boucheron2015tail} impose that the survival function $\bar F$ (or equivalently a quantile function) admits a standardized Karamata representation with a bias function (also called von Mises function) decreasing faster than a  certain  power of the quantile level 
(see Condition~\ref{cond:vonMises} in 
Section~\ref{sec:secondorder}). This implies in particular the condition required by \cite{carpentier2015adaptive}. 
 Further discussion regarding different types of second-order assumptions required in the literature is deferred to Section~\ref{sec:secondorder}. Additionally, it is important to note that the analysis carried out in \cite{boucheron2015tail} involves critical tuning parameters, which are chosen in their experiments through preliminary calibration. With these carefully tuned parameters, the authors show through simulation studies that their adaptive rule attains similar performance as the one proposed by \cite{drees1998selecting}, under weaker second order conditions. 

Our contribution in this work is to provide a simple and transparent
adaptive method for tail index estimation, not requiring any
calibration nor relying on second order or von Mises assumptions that
can be difficult to verify.  Our approach is inspired by an adaptive strategy 
called `Adaptive Validation' (AV)~\citep[Chapter~4]{Lederer2021HD}, which
itself is inspired by Lepski's method in non-parametric
regression~\citep{lepskiui1990problem,lepski1997optimal,goldenshluger2011bandwidth}. Lepski's method has inspired a variety of adaptations to specific settings in statistics 
\citep{comte2013anisotropic,lacour2016minimal}.
The AV approach primarily requires a sharp control of
a variance term within an error decomposition, without imposing any
particular condition on the bias term except that it can be made
smaller than the variance term for certain values of the tuning
parameter (here, for small~$k$).  To date, these general tools have
not been considered within the field of EVT. Under the minimal regular
variation assumption~(\ref{eq:rv-definition}), we  propose an `Extreme Adaptive Validation' (EAV) rule for selecting $k$, denoted as $\kAV$. We demonstrate that $\kAV$  performs comparably to a certain ‘oracle’ choice
of $k$, denoted as $\kS$, for which we derive some optimality properties (see Section~\ref{sec:AVframework}).  Furthermore, under
additional von Mises Condition~\ref{cond:vonMises}, we establish that
$\kS$  achieves minimax optimal error rates, up to a power of $\log\log n
$ and we show that the adaptive
validation method $\kAV$ attains nearly minimax rates up to a nearly constant factor which is a power of $\log \log n$.  Also instead of considering all possible values of $k$, our approach restricts to a
much smaller logarithmic grid, a strategy commonly used in
practice. 
This significantly reduces computational complexity: while
previous methods have worst-case complexity of $O(n^2)$ and a
posteriori complexity of $O(\hat{k}^2)$, our approach achieves
$O((\log n)^2)$ complexity, and in practice, a posteriori complexity
of only $O(\log(\kAV)^2)$. 
This reduction in complexity is not crucial for estimators that are not computationally intensive. However, it is anticipated that this simplification may prove beneficial in future works extending the EAV framework to more complex tasks, such as multivariate extreme value analysis.

The paper is organized as follows:
In Section~\ref{sec:adpv}, we describe the general principles underlying an EAV  rule, along with its guarantees, at a rather high level of generality, as the described methodology is applicable to any estimator of the tail index that satisfies a certain bias-variance decomposition. In Section~\ref{sec:BVD}, we specialize to the case of the Hill estimator, for which we derive non-asymptotic error bounds, allowing us to leverage the results of Section~\ref{sec:AVframework}. In Section~\ref{sec:secondorder}, we consider the case where additional second-order conditions hold in addition to Equation~\eqref{eq:rv-definition}, and we provide additional guarantees in this setting. We support our theoretical findings with simulations in Section~\ref{sec:sim}. Additional theoretical details and simulations are deferred to the Appendix.

%% file: Contents2/SectionAV.tex
\section{Adaptive validation}\label{sec:adpv}
We begin by introducing the Extreme Adaptive Validation (EAV) procedure that we
promote, highlighting its broad applicability.  Our analysis is
applicable to any estimator $\hat{\gamma}(k)$ of $\gamma$ that takes
as input the $k$ largest order statistics of an i.i.d. sample drawn
from $F$, provided that non-asymptotic error bounds are known and take
the form of a bias-variance decomposition satisfying the following
condition.
\begin{condition}[Bias-Variance Decomposition]\label{cond:genericErrorDecomp}
  There exists two functions $V:\nset\times (0,1)\to\rset_+$ and $B:\nset\times \nset\times (0,1)\to \rset_+$, 
  such that for any confidence level $\delta\in(0,1)$, 
  and any 
  $k\in\{1,\dots,n\}$, 
  with probability at least  $1-\quantile$,
  \begin{equation}
    \label{eq:assumptionGamma}
    |\hat \gamma(k) - \gamaOrc| ~ \le ~\gamaOrc ~\varianceF{k,\quantile} ~+~
    \biasF{ k,n, \quantile}, 
  \end{equation}
where the functions  $\variance(\cdot,\cdot)$,  $\biasF{\cdot,\cdot, \cdot}$ satisfy the following requirements: 
  \begin{enumerate}[(a)]
  \item The term $\varianceF{k,\quantile}$ has an explicit expression,
    and $\varianceF{\point,\quantile}$ is a non-increasing function of
    its first argument $\indexk$ such that $\varianceF{k,\delta}\to 0$
    as $k\to\infty$.
  \item The function $\bias(\cdot,\cdot,\delta)$ is non-negative and
    non-decreasing in its first argument $k$. While it may not have an
    explicit expression, it satisfies $B(k,n,\delta) \to 0$ as
    $k/n \to 0$.
  \end{enumerate}
\end{condition}
\noindent In Condition~\ref{cond:genericErrorDecomp}, the function $V$
should be seen as a deviation term associated with the variability of
averages over the $k$ largest order statistics. The larger $k$, the
smaller $V$. The rationale behind the factor $\gamma$ in the required
error decomposition~\eqref{eq:assumptionGamma} is that the standard
deviation of the Hill estimator is $\gamma/\sqrt{k}$ in the ideal case
of a Pareto distribution.  The function $B$ should be seen as a bias
term: the larger $k$, the less extreme the considered data are, the
larger $B$. We emphasize that explicit knowledge of $B$ is not
required.  Note that the requirement in (b) that $\bias$ is a
non-decreasing function of $k$ is automatically satisfied when
replacing $\bias(k,n,\quantile)$ with
$\tilde{ \bias}(k,n,\quantile) = \sup_{k'\le k}
\bias(k',n,\quantile)$. 
We shall prove in Section~\ref{sec:BVD} that the Hill estimator
$\HillEst$ indeed satisfies Condition~\ref{cond:genericErrorDecomp}
with an explicit deviation term that is a quantile of a centered Gamma
random variable of order
$V(k,\quantile)\le \sqrt{2\log(4/\quantile)/k} + 2\log(4/\quantile)/k
$, and a bias term that is not explicit but satisfies requirement (b)
in Condition~\ref{cond:genericErrorDecomp}.

\subsection{Adaptive Validation framework}\label{sec:AVframework}
We are now prepared to introduce an Adaptive Validation (AV) approach
for selecting the extreme sample size $k$ with statistical guarantees,
building upon the methodologies developed
in~\cite{chichignoud2016practical,li2019tuning,Taheri2019,laszkiewicz2021thresholded}. These
works have set forth algorithms for calibrating tuning parameters in
diverse contexts, including standard linear and logistic regression,
as well as graphical models, outside the EVT framework.

In this section, we refrain from imposing additional assumptions (such
as second-order or von Mises conditions). Consequently, our objective
is not to identify an optimal $k$, but rather a `good enough' $k$ that
comes with certain guarantees. Let $\hat{\gamma}$ satisfy
Condition~\ref{cond:genericErrorDecomp}. For any confidence level
$\delta$, we define an `oracle' $\kS(\delta,n)$ as one that balances
$\variance(k,\delta)$ and $\bias(k,n,\delta)$
in~\eqref{eq:assumptionGamma}, among candidates $k$ chosen from a
grid $\grid \subset \{1, \ldots, n\}$ of size $|\grid| \leq n$. In
this setting, 
we (naturally) need to assume that the grid 
 is chosen `wide' enough to contain a $k$ reaching a bias-variance
 compromise, as encapsulated in the following condition:
\begin{condition}[Sufficiently wide grid]\label{cond:wideGrid}
  The triplet $(\delta,n,\grid)$ is such that with $k_{\min} = \min(\grid)$, $k_{\max} = \max(\grid)$,
  \begin{equation*}
  \biasF{k_{\min}, \nbrsamp,\quantile} \le \gamaOrc \varianceF{k_{\min},\quantile}  ~\text{and}~
     \biasF{k_{\max}, \nbrsamp,\quantile} > \gamaOrc \varianceF{k_{\max},\quantile}. 
  \end{equation*}
\end{condition}
Note that Condition~\ref{cond:wideGrid} is automatically satisfied for grids such that $k_{\min}=1,k_{\max}=n$, for $n$ large enough, in view of Condition~\ref{cond:genericErrorDecomp}. 
For $(\delta,n,\grid)$ satisfying
Condition~\ref{cond:wideGrid} we define the oracle $\kS(\delta,n)$ as
follows:
\begin{equation}
  \label{eq:kstar}
  \kS(\quantile,n) ~=~ \max\{\indexk\in \grid:
  \biasF{\indexk, \nbrsamp,\quantile} \le \gamaOrc \varianceF{\indexk,\quantile} \}.
\end{equation}

Note that under Condition~\ref{cond:wideGrid}, the oracle \(\kS(\delta, n)\) is well-defined and \(k_{\min} \le \kS(\delta, n) < k_{\max}\). It may seem surprising and perhaps not ideal that the definition of the oracle depends on the choice of the grid. One might argue that this approach simply replaces the choice of \(k\) with the choice of \(\grid\). However, there are universal default choices for \(\grid\), such as geometric grids or linearly spaced grids.
Moreover, with sufficiently fine grids that have a high enough maximum
and logarithmic grid size,
\(|\grid| \propto \log(n)\),  we can show that the oracle error
achieves known minimax rates under second-order conditions.
Further discussion is deferred to Section~\ref{sec:secondorder}.

Informally, the oracle $\kS(\delta,n)$ ensures a compromise between
bias and variance, in the sense that 
$\gamaOrc\varianceF{ k, \quantile}\approx \biasF{k,\nbrsamp, \delta}
$.  Formally, we obtain immediately an upper bound for the error of
$\kS(\delta,n)$, without a bias term.  The following result derives
immediately from the very definition of $\kS$, from the upper
bound~\eqref{eq:assumptionGamma} and from the monotonicity
requirements on the functions $V$ and $B$.
\begin{proposition}[Explicit error bound for the oracle]\label{prop:errorOracle}
  Let $\hat \gamma$ satisfy
  Condition~\ref{cond:genericErrorDecomp}. For any $\delta\in(0,1)$,
  $n\ge 1$ and $\grid$ satisfying Condition~\ref{cond:wideGrid}, we
  have for all fixed $k ~\le~ \kS(\delta,n)$, with probability at
  least $1-\delta$, 
\begin{equation}
  \label{eq:oracleBound}
  |\gamaOrc - \HillEst(k)| \le 2 \gamaOrc \varianceF{k,\quantile}. 
\end{equation}    
\end{proposition}
\noindent Note that the `price' paid by the oracle $\kS(\delta,n)$ to benefit
from a clear guarantee without a bias term (in the ideal case where
the value of $\kS(\delta,n)$ were known) is relatively small. Indeed
the variance term in the upper bound is only multiplied by a factor
$2$. The upper bound on the oracle error, as stated in
Proposition~\ref{prop:errorOracle}, is derived under
Condition~\ref{cond:genericErrorDecomp}, where the function $V$ is
assumed to be explicit. However, since $\kS(\delta,n)$ itself is
unknown, this bound is not fully explicit unless a lower bound on
$\kS(\delta,n)$ can be derived. We conjecture that there is no
universal method to achieve this without further assumptions. In
Section~\ref{sec:secondorder}, we obtain an explicit lower bound for
$k^*$ for the Hill estimator under additional von Mises conditions.

The following result shows that the oracle $\kS$ enjoys
approximate optimality guarantees under minimal assumptions:
 \begin{proposition}[Optimality properties of $\kS$]\label{prop:optimality_kS}
   For fixed $\delta>0, n\ge 1$, $\grid$ satisfying Condition~\ref{cond:wideGrid}, 
    let
   $$E(k) ~=~ \gamma V(k,\delta) + B(k,n,\delta)$$
   denote the upper bound on the error of an estimator $\hat \gamma$
   satisfying Condition~\eqref{cond:genericErrorDecomp}.  Then for all
   $1\le k\le \kS(\delta,n)$,
$$
 E(k) ~\le~ \min_{1\le j\le n} E(j) + \gamma V(k,\delta).
$$
\end{proposition}
\begin{proof}
  Let $k\le \kS(\delta,n)$.  For simplicity of notation, we abbreviate \( V(k, \delta) \) as 
  \( V(k) \), and
  \( B(k, n, \delta) \) as \( B(k) \).  For $j \ge k$, we have
  \begin{align*}
    E(j) ~ \ge ~B(j) ~\ge B(k) ~ = ~ E(k)-\gamma V(k).
  \end{align*}
  On the other hand notice that for any $k\le \kS$ we have
  $\gamma V(k) \ge \gamma V(\kS)\ge B(\kS)\ge B(k)$ by monotonicity of
  the functions $V,B$. Thus for $1\le j\le k\le \kS$, 
  \begin{align*}
    E(j) & ~= B(j) + \gamma V(j) ~\ge ~
           B(j) + \gamma V(k) \ge 
           B(j) + B(k) 
    ~\ge ~B(k) ~=~ E(k)- \gamma V(k).
  \end{align*}

\end{proof}

\begin{remark}[Interpretation]\label{rem:interpretation_optimality} 
  Proposition~\ref{prop:optimality_kS} implies that for
  $k \leq \kS(\delta,n)$ 
  the upper bound
  $E(k)$ on the error of $\hat{\gamma}(k)$ is optimal, up to an
  additional `regret' term $V(k,\delta)$. This regret term is
  minimized when $k = \kS(\delta,n)$.

  Another interpretation is as follows: $\kS$ minimizes the upper
  bound $2\gamma V(k,\delta)$ over the set of indices $k$ for which
  this upper bound (which is the only possible explicit upper bound,
  given that $B$ is unknown and only $V$ is known) is valid:
$$
\kS(\delta,n) ~\in ~\argmin_{k \in \grid: E(k) \leq 2\gamma V(k,\delta)} 2\gamma V(k,\delta).
$$
\end{remark}

We now consider an empirical version of $\kS$ using an Adaptive
Validation rule, akin to those in \cite{boucheron2015tail} and
\cite{drees1998selecting}. Our approach distinguishes itself by
providing an explicit stopping rule, unlike \cite{boucheron2015tail},
and guarantees valid for any sample size, unlike
\cite{drees1998selecting}. These guarantees, derived in
Section~\ref{sec:guarantees}, ensure that the error from the adaptive
rule is comparable to the oracle error induced by $\kS$.
We first need to define minimal admissible sample sizes depending on
the confidence level $\delta$, for technical reasons which have an
intuitive interpretation. 
First, because our proof techniques employ union bounds on the
probability of adverse events defined for any $k \in \grid$, our
analysis will repeatedly involve a tolerance level
$\deltaG:=\delta/|\grid|$.  Alternative approaches involving chaining
techniques are discussed in Remark~\ref{rem:chaining}.

Second, we need to consider only extreme sample sizes $k$ such that the
error due to variance is less than $1/2$. This is because our rule
involves division by $1 - 2V(k, \deltaG)$. We thus 
restrict the search to $k \geq k_0(\delta)$, where $k_0(\delta)$
satisfies
\begin{equation}\label{eq:defk0}
k_0(\delta) ~\ge~ \inf \{k\ge 1:  V(k,\deltaG) < 1/2 \}. 
\end{equation}
%
%

Finally, given a choice of $k_0(\delta)$ satisfying~\eqref{eq:defk0},
we must ensure that the bias term is indeed less than the variance
term for $k = k_0(\delta)$, so that $\kS(\deltaG, n)$ is not less than
$k_0(\delta)$. We thus define the minimum (full) sample size as
\begin{equation}\label{eq:defn0}
  \begin{aligned}
    n_0(\delta) 
                & ~=~  \inf\big\{~n: ~  k_0(\delta) ~\le ~ \kS(\deltaG,~n)~ \big\}, 
  \end{aligned}
 \end{equation}
  Note that if the triplet
 $(\deltaG,n,\grid)$ satisfies Condition~\ref{cond:wideGrid}, then
 $n_0(\delta) = \inf\big\{~n : ~\biasF{k_0(\delta), n, \deltaG} ~ \le \allowbreak
 ~ \gamma V\big(k_0(\delta),\deltaG\big) ~\big\}$ and although
 $n_0(\delta)$ is unknown, it is guaranteed to be a finite
 integer 
 due to our requirement that $B(k,n,\delta) \to 0$ as $n \to \infty$
 for fixed $k$ (Requirement (b) in
 Condition~\ref{cond:genericErrorDecomp}).
In Section~\ref{sec:secondorder}, we derive an explicit expression for $n_0(\delta)$ under additional von Mises conditions in the context of Hill estimation. This expression can serve as a guideline to assess whether the available sample size is sufficient for our theory to apply (see Corollary~\ref{cor:minsamplesize}).

We are now ready to define an \emph{Extreme Adaptive Validation (EAV)}
sample size $\kAV$ chosen by adaptive validation as follows,  for any pair
$(\delta,n)\in(0,1)\times\nset$ such that 
$n\ge n_0(\delta)$, 
\begin{equation}
  \label{eq:khat-gamma}
\begin{aligned}
\kAV ~=~ \max \Bigl\{
&k \in \grid, \, k \geq k_0(\delta), \text{ such that } \forall j \in \grid \text{ with } j \leq k, \\
&\left| \HillEst(k) - \HillEst(j) \right| ~ \leq ~ \frac{\HillEst(k)}{1 - 2 \variance(k, \deltaG)} \bigl( \variance(j, \deltaG) + 3 \variance(k, \deltaG) \bigr)
\Bigr\}.
\end{aligned}
\end{equation}
It will be useful in subsequent analysis to introduce the binary
random variable associated with our stopping criterion, for
$k_0(\delta)\le k\le n$,
\begin{equation}
  \label{eq:criterion}
  \begin{aligned}
    S(k) ~=~ &\un\Bigl\{\exists j\in \grid, j\le k, \text{ such that } \\
    &| \HillEst(k) - \HillEst(j)| >
  \frac{\HillEst(k) }{1- 2 \variance(k,\deltaG)} \bigl(\variance(j,\deltaG) + 3 \variance(k,\deltaG)\bigr)
  \Bigr \}\,.    
  \end{aligned}
\end{equation}
Thus, if $\grid = \{k_m, m\le |\grid|\}$ and
$\kAV = k_{\widehat m } < \max(\grid)$ then
$S(k_{\widehat m + 1}) = 1$ while $S(k_{\widehat m}) = 0$, with
probability one.  Implementing $\kAV$ can be achieved by a
straightforward algorithm taking as input
$(n,\delta, k_0(\delta), \grid)$ and computing $S(k_m)$ for increasing
values of $k_m$ within the range
$ \grid\cap \{k_0(\delta),\ldots, n\}$. The algorithm stops when
$S(k_m)=1$ and returns $\kAV = k_{m-1}$ where $m$ is the current index
value on the grid $\grid$. Of course there is the theoretical
possibility that $S(\max(\grid)) = 0$ in which case the algorithm
should return $\kAV=\max(\grid)$. This latter event is however
unlikely if the upper end-point $k_{\max}$ of the grid is chosen large
enough. 

 \subsection{Statistical guarantees for EAV}\label{sec:guarantees}

 We now derive guarantees for the adaptive index $\kAV$ proposed
 in~\eqref{eq:khat-gamma}.  Our main result
 (Theorem~\ref{theo:AVguarantee} below) shows that the error of
 $\hat{\gamma}(\kAV)$ is of the same magnitude as the oracle error
 $|\hat\gamma(\kS(\deltaG,n)) - \gamma|$.  As a first go, we
 state convenient uniform bounds on the error over candidates $k$ in
 $\grid$ such that $k\le \kS$(Recall from~\eqref{eq:kstar} the
 definition of the oracle $\kS$). 
\begin{lemma}\label{lem:UnifDev-eventE}
  Let $\hat \gamma$ be an estimator of $\gamma$ satisfying
  Condition~\ref{cond:genericErrorDecomp}. Let $\delta\in(0,1)$, 
  $n\in\nset$ and $\grid$ be such that $n\ge n_0(\delta)$ and the
  triplet $(\deltaG,n,\grid)$ satisfies Condition~\ref{cond:wideGrid}.
  On an event $\mathcal{E}$ of probability at least
  $1-\delta$, we have for all $k \in \grid$ such that $k\le 
  \kS(\deltaG,n)$, 
\begin{equation}\label{eq:unifDevHatGamma}
  |\hat \gamma (k) - \gamma|
 ~\le~  \gamma V(k,\deltaG) + B(k,n,\deltaG)
 ~ \le ~2 \gamma V(k,\deltaG). 
\end{equation}
On the same  event, for all $k\in \grid$ such that $k_0(\delta)\le
k\le \kS(\deltaG,n)$, we have
 \begin{equation}\label{eq:boundGamma}
\gamaOrc\le \hat\gamma(k)/(1 - 2 \variance\bigl(k,\deltaG)\bigr). 
 \end{equation}

\end{lemma}
\begin{proof}
  Condition~\ref{cond:genericErrorDecomp}  combined with a union bound ensures that on an event $\mathcal{E}$ of probability at least $1- \delta$,  
  Inequality~\eqref{eq:assumptionGamma}  is satisfied with $\delta$ replaced with $\deltaG$ for all
  $k \in\grid$.  
  Thus, on the latter  event $\mathcal{E}$, it holds that 
  for all $k \in \grid\cap\{1,\ldots,  \kS(\deltaG, n) \}$, 
  $|\hat \gamma (k) - \gamma|\le \gamma V(k,\deltaG) + B(k,n,\deltaG)
  $. This proves the first inequality
  in~\eqref{eq:unifDevHatGamma}. The second inequality
  in~\eqref{eq:unifDevHatGamma} derives immediately from the
  definition of the oracle $\kS$
  in~\eqref{eq:kstar}. 
    
    Also on $\mathcal{E}$, it holds that for all $k\in\grid$ such that
    $k\le
    \kS(\deltaG,n)$, 
 $$
\gamaOrc  ~\le ~\HillEst(k) + \gamaOrc \variance(k,\deltaG) + \bias(k, n, \deltaG) \le \HillEst(k) + 2 \gamaOrc \variance(k,\deltaG)\,.     
   $$
   Finally for  $k \ge k_0(\delta)$,  
   we have $ V(k,\deltaG) \le 1/2$, and  the latter display yields~\eqref{eq:boundGamma} by a straightforward inversion. 
  \end{proof}

  
 The following result shows that the algorithm described above for computing $\kAV$ does not stop too early. 
 \begin{proposition}[$\kAV\ge \kS(\deltaG,n)$ with high
   probability]\label{prop:hatk}
In the setting of Lemma~\ref{lem:UnifDev-eventE}, 
and on the same event $\mathcal{E}$, for all
$(j,k)\in\{1,\ldots,\kS(\deltaG,\,n)\}$ 
such that $j\le k$ and $k\ge k_0(\delta,n)$, we
have \begin{equation}\label{eq:mainResultDeviations} |\HillEst(k) -
  \HillEst(j)| ~ \le ~\frac{\HillEst(k)}{1 - 2 \variance(k,\deltaG)}
  \bigl(3\variance(k,\deltaG) + \variance(j,\deltaG)\bigr).
\end{equation}  
As a consequence, for fixed $(\delta,n)$  the adaptive  $\kAV$ defined in~\eqref{eq:khat-gamma} satisfies 
  $$
\kAV ~ \ge ~\kS(\deltaG, \, n).
  $$
\end{proposition}

  \begin{proof}[of Proposition~\ref{prop:hatk}]

  On the   event $\mathcal{E}$ from Lemma~\ref{lem:UnifDev-eventE}, we have 
  for all  $j,k\in\grid$ such that $ j\le k\le \kS(\deltaG,\, n) $, 
    \begin{align}
      |\HillEst(k) - \HillEst(j)|
      &~\le~ |\HillEst(k) - \gamaOrc| + |\HillEst(j) - \gamaOrc| \nonumber\\
      &~\le~ \gamaOrc \bigl(\variance(k,\deltaG) + \variance(j,\deltaG) \bigr) +
        \bias(k,\nbrsamp,\deltaG)+\bias(j,\nbrsamp,\deltaG) \text{ \quad(from~\eqref{eq:unifDevHatGamma})} \nonumber\\
      & ~\le~ \gamaOrc \bigl(\variance(k,\deltaG) + \variance(j,\deltaG) \bigr) + 2 \bias(k,\nbrsamp,\deltaG)  \text{\quad (from Condition~\ref{cond:genericErrorDecomp}(b))}\nonumber\\
      &~\le~ \gamaOrc \bigl(3\variance(k,\deltaG) + \variance(j,\deltaG)\bigr)
      \text{ \quad(from~\eqref{eq:kstar})}\,.
      \label{eq:boundAllDevEmp}
    \end{align} 
    Combining~\eqref{eq:boundAllDevEmp} and~\eqref{eq:boundGamma}
    yields that on $\mathcal{E}$, for $j,k\in\grid$ such that
    $k_0(\delta)\le k \le \kS(\deltaG,n)$ and $1\le j\le k$,
    \eqref{eq:mainResultDeviations} holds. Writing
    $\kAV = k_{\widehat m}$ and $\kS(\deltaG, n) = k_{m^*} $, we
    obtain that on $\mathcal{E}$, $S(k_m) = 0$ for all $m\le
    m^*$. However from the definitions, eiher $ \kAV = \max(\grid)$
    or $S(k_{\widehat m + 1})=1$, thus $\widehat m + 1 > m^*$, from
    which it follows that $\widehat m\ge m^*$ and finally
    $\kAV\ge \kS(\deltaG, n)$.
    
    \end{proof}

    \begin{theorem}[Error bounds for
      $\hat\gamma(\kAV)$]\label{theo:AVguarantee}
      In the setting of
      Lemma~\ref{lem:UnifDev-eventE} 
      and on the same favourable event $\mathcal{E}$ of probability at
      least $1-\delta$,
the adaptive  estimator resulting from the AV rule satisfies
  \begin{equation*}
|\hat \gamma (\kAV) - \gamma| ~\le~ 
%
\Big( \frac{4\hat \gamma(\kAV)}{1 - 2V(\kAV, \deltaG)} + 2\gamma \Big)~
V\bigl(\kS(\deltaG,\,n),\, \deltaG \bigr) .  
  \end{equation*}
  Assuming $V(\kS(\deltaG,\,n),\, \deltaG) < 1/6$,  the above inequality implies the simplified bound 
  \begin{equation*}
 |\hat \gamma (\kAV) - \gamma| ~\le~      
 \frac{6\gamma}{1 - 6 V\bigl(\kS(\deltaG,\,n),\, \deltaG \bigr)}~
  V\bigl(\kS(\deltaG,\,n),\, \deltaG \bigr)\,.
  \end{equation*}
\end{theorem}
We provide the proof below after a few
comments. Theorem~\ref{theo:AVguarantee} establishes an explicit upper
bound on the error of $\hat{\gamma}(\kAV)$ in terms of $\kS$, under
the minimal Condition~\ref{cond:genericErrorDecomp}.  As with the oracle error
(Proposition~\ref{prop:errorOracle}), the bound in
Theorem~\ref{theo:AVguarantee} depends on the unknown magnitude of
$\kS$. To address this, in Section~\ref{sec:secondorder}, we establish
a lower bound for $\kS$ under the additional von Mises
Condition~\ref{cond:vonMises}.


\begin{remark}[Comparison With The Oracle Error]  
  With reasonably large sample sizes, the variance term
  $V(\kS(\deltaG,n),\deltaG)$ should be small. Also from
  Proposition~\ref{prop:hatk} and
  Condition~\ref{cond:genericErrorDecomp}, we have
  $V(\kAV,\deltaG )\le V(\kS(\deltaG,n),\deltaG)$, so that
  $V(\kAV,\deltaG )$ should also be small. In this setting, the upper
  bound in Theorem~\ref{theo:AVguarantee} becomes
  approximately \begin{equation*} |\hat \gamma (\kAV) - \gamma|
    \lesssim 6 \gamma V\bigl(\kS(\deltaG,n), \deltaG\bigr)
    , \end{equation*} which is only three times the oracle error
  stated in Proposition~\ref{prop:errorOracle}, up to replacing
  $\delta$ with $\deltaG$.  For the Hill estimator, 
  the leading term of \( V \) depends on \( \delta \) only
  through a logarithmic term \( \sqrt{\log(1/\delta)}
  \). Consequently, the difference between the errors of the oracle
  \( \kS(\delta, n) \) and \( \kS(\delta_{\grid}, n) \) is
  merely a factor of \( \sqrt{\log|\grid|} \), meaning that for
  grids of logarithmic size, this difference is
  \( O(\sqrt{\log \log n}) \).  Thus, the adaptive extreme sample size
  \( \kAV \) is almost optimal in the sense that its associated error
  is bounded by a multiple of the error bound associated with the
  (unknown) oracle \( \kS \) up to a quasi-constant factor. This
  oracle is itself optimal in a certain sense, as detailed in
  Proposition~\ref{prop:optimality_kS} and
  Remark~\ref{rem:interpretation_optimality}.
\end{remark}
\begin{proof}[of Theorem~\ref{theo:AVguarantee}]
  For simplicity of notation, we abbreviate $\kS = \kS(\deltaG,n)$, $V(k)=V(k,\deltaG)$, and $B(k)=  B(k,n,\deltaG)$
  in this proof.
We first decompose the error as 
\begin{equation}\label{eq:decomposErrKav}
    |\hat \gamma(\kAV) - \gamma|
      ~\le~   |\hat \gamma(\kAV) - \hat \gamma(k^*)| +
      |\hat \gamma(k^*) - \gamma|.  
\end{equation}
By construction,  $S(\kAV)=0$ (see~\eqref{eq:criterion} and the discussion below). Also, on  the favourable event $\mathcal{E}$ introduced in Lemma~\ref{lem:UnifDev-eventE},  we have from Proposition~\ref{prop:hatk} that $\kS \le \kAV$.  In addition  both $\kS$ and $\kAV$ belong to the grid $\grid$ where the stopping criterion in~\eqref{eq:khat-gamma} is evaluated,  so  the first term in the right-hand side of~\eqref{eq:decomposErrKav} must satisfy
\begin{align}
    |\hat \gamma(\kAV) - \hat \gamma(k^*)| 
   ~\le ~ \frac{\hat \gamma(\kAV)}{1 - 2V(\kAV)} \bigl( 3V(\kAV)+V(k^*)\bigr )
   ~ \le ~\frac{4 \hat \gamma(\kAV) V(k^*)}{1 - 2V(\kAV)}\,,
   \label{eq:firstterm} 
\end{align}
where we have used (see Condition~\ref{cond:genericErrorDecomp}(a))
that $V$ is a non-increasing function of
$k$ to obtain the latter inequality.  On the other hand from
Lemma~\ref{lem:UnifDev-eventE}, the second term
in~\eqref{eq:decomposErrKav} satisfies on $\mathcal{E}$,
\begin{equation}\label{eq:boundDevkstar}
|\hat \gamma(\kS) - \gamma|\le 
\gamma V(\kS) +  B(\kS) 
~\le~ 2 \gamma V(\kS).  
\end{equation}
Combining~\eqref{eq:decomposErrKav}, \eqref{eq:firstterm} and~\eqref{eq:boundDevkstar} ends the proof of the first claim.

Now, using the fact that $\kAV\ge \kS$ on $\mathcal{E}$, together with Condition~\ref{cond:genericErrorDecomp}(a), and the triangle inequality, the first claim implies that 
 \begin{align*}
|\hat \gamma (\kAV) - \gamma|  
~\le~ \biggl( \frac{4\hat \gamma(\kAV)}{1 - 2V(\kS)} + 2\gamma \biggr)~V(\kS) 
~\le~ \biggl( \frac{4|\hat \gamma(\kAV)-\gamma|+4\gamma}{1 - 2V(\kS)} + 2\gamma \biggr)~V(\kS), 
  \end{align*}
so that 
 \begin{align*}
\biggl(1-\frac{4 V(\kS) }{1 - 2V(\kS)}\biggr)|\hat \gamma (\kAV) - \gamma| & ~\le ~
 \biggl( \frac{4\gamma}{1 - 2V(\kS)} + 2\gamma \biggr) V(\kS). 
  \end{align*}  
Using that  $V(\kS) < 1/6$ we obtain 
  \begin{align*}
    |\hat \gamma (\kAV) - \gamma|
    &~\le ~\biggl( \frac{4\gamma}{1 - 2V(\kS)} + 2\gamma \biggr)
      \biggl(1-\frac{4 V(\kS) }{1 - 2V(\kS)}\biggr)^{-1}  V(\kS) \\
    &~= ~ \gamma\big( 6-4 V(\kS)  \big)
      \big(1 - 6 V(\kS) \big)^{-1} ~  V(\kS) \\
    &~=~
      \frac{ 6\gamma}{1 - 6 V(\kS) }~ V(\kS)\,,
  \end{align*}  
 as desired. 
  
\end{proof}

%% file: Contents2/SectionHill.tex
\section{Bias-Variance decomposition for the Hill estimator}\label{sec:BVD}
We now conduct a non-asymptotic analysis of the Hill estimator. We impose the—very weak—assumption that the observations \( X_1, \ldots, X_n \) are independent samples with a survival function \( \bar{F} \) satisfying the regular variation condition~\eqref{eq:rv-definition}. A key result is a bias-variance decomposition of the error that meets exactly the requirements of Condition~\ref{cond:genericErrorDecomp} (Theorem~\ref{thm:main}). Our main result (Corollary~\ref{cor:hatk}) follows directly from this, demonstrating that our general analysis in Section~\ref{sec:AVframework} applies to the Hill estimator.
\subsection{Karamata representation of the Hill estimator}\label{sec:decomposeHill}
We first recall basic facts about regularly varying
functions~\citep{binghametal1987,de1987regular} and introduce some notation. 
A function $H:\rset_+\to\rset_+$ is called 
regularly varying with index ${r}\in\rset$, if 
$H(tx)/H(t)\to x^{r}$ as 
$t\to\infty$, for all $x>0$, as in \eqref{eq:rv-definition}. The exponent~${r}$ is the regular variation index and~$H$ is called `slowly varying' if
${r}=0$. Now, any regularly varying function $H$ with index ${r}$ can
be written as $H(x) = L(x) x^{{r}}$ with $L:\rset_+\to\rset_+$ 
slowing varying. Finally, a key observation is that if $\bar F$ is regularly varying with index $-\alpha <0$
then also the quantile function $Q$ defined as
$Q(t)=F^{\leftarrow}(1-1/t)$ is regularly varying with positive index
$  \gamma = 1/\alpha$. Note that the usual notation for our function $Q$ is $U$. We prefer not using $U$ in order to avoid confusion with uniform order statistics which play a central role in this section.  The next representation of slowly varying function
is a famous consequence of Karamata's Tauberian theorem, and it is key to
understand the deviations of the Hill
estimator. 
A function $L$ is slowly varying if and only if there exists two functions 
$b:\rset_+\to\rset$ and $a: \rset_+\to \rset_+$ such that
$\lim_{t\to\infty} b(t) = 0$ and $\lim_{t\to \infty} a(t)= A\ge 0$,
such that for some $t_0\ge 0$ and all $t\ge t_0$, 
  \begin{equation}\label{eq:karamata}
    L(t)~=~a(t) \exp\int_{t_0}^t \frac{b(u)}{u} \ud u\,.
  \end{equation}
  Summarizing, for $\bar F$ satisfying~\eqref{eq:rv-definition} and
  $Q$ the quantile function of $F$ as above, we may write
  $Q(t) = L(t) t^{\gamma}$, where $L$ is as
  in~\eqref{eq:karamata}.  
  Define
  \begin{equation}
    \label{eq:barabarb}
   \bar b(t)~=~\sup_{x\ge t} |b(x)|, \qquad \bar a(t)~=~\sup_{x\ge t} |a(x) - A|\,. 
  \end{equation}
 Then both functions
  $\bar b$ and $\bar a$ are monotonically non-increasing and both
  converge to $0$.
  \begin{remark}[Comparison with~\cite{boucheron2015tail}]
    The von Mises assumption made in \cite{boucheron2015tail} is that the quantile function $Q$ writes $Q(t) = t ^\gamma L(t)$ where $L$ satisfies~\eqref{eq:karamata} with a constant function $a(t)=A$, so that $\bar a \equiv 0$. In addition, lower bounds on the Hill estimator are obtained under the additional assumption that the so-called von Mises function $\bar b$ (denoted by $\bar \eta$ in the cited reference) satisfies $|\bar b(t) |\le Ct ^\rho$ for some $C>0$ and $\rho<0$.   Note that our guarantees on the Hill estimator in this section do not require such second order assumptions, although  we consider it  (Condition~\ref{cond:vonMises}) in Section~\ref{sec:secondorder}.
  \end{remark}
  Let $U_{(1)}\ge \dots \ge U_{(n)}$ denote  the order statistics of an independent sample of standard uniform variables. 
  Then, $\HillEst \eqd \sum_{i=1}^k\log Q[(1 - U_{(i)})^{-1}] - \log Q[(1 - U_{(k+1)})^{-1}]$, which by~\eqref{eq:karamata} yields 
  \begin{align}
    \HillEst(k)
    \qquad\eqd \qquad &  
       \frac{\gamaOrc}{k}\sum_{i=1}^k \bigl( -\log(1 - U_{(i)}) + \log(1 - U_{(k+1)})  \bigr) \quad + \dotsb  \label{eq:exposum}  \\
   & \frac{1}{k} \sum_{i=1}^k \Bigl(a\big(( 1 - U_{(i)})^{-1} \big) - a\big(( 1 - U_{(k+1)})^{-1} \big)\Bigr) \quad +\dotsb \label{eq:bias1} \\
&    \frac{1}{k} \sum_{i=1}^k \int_{( 1 - U_{(k+1)})^{-1} }^{( 1 - U_{(i)})^{-1} }
    b(u)/u \; \ud u\,, 
\quad  \label{eq:bias2} 
  \end{align}
  where the functions $a$ and $b$ are the ones appearing in the Karamata representation of $Q$. The decomposition in the above display is well known can be found \emph{e.g.} in  \citet[Proof of Proposition~1(E)]{mason1982laws}.
Introduce the random variable
\begin{equation}
  \label{eq:defZk}
\ZK~=~\frac{1}{k}\sum_{i=1}^k \bigl( -\log(1 - U_{(i)}) + \log(1 - U_{(k+1)})\bigr),  
\end{equation}
so that the  first term~\eqref{eq:exposum} above writes  $\gamma Z_k$.
 A second important fact   also shown in~\cite{mason1982laws}, relying on Renyi representation of exponential spacings, 
 is that
 \begin{equation}
  \label{eq:identityExposum}
\ZK  \eqd  \frac{1}{k}\sum_{i=1}^k E_i,  
\end{equation}
where $(E_i,i\le n)$ are independent unit exponential random variables. Thus,  $\ZK$ follows a Gamma distribution $\mathcal{G}(\alpha,\beta)$ with shape and rate parameter, respectively $\alpha= k, \beta=k$. 

  \subsection{Error decomposition} 
  Equipped with the Karamata representation of the Hill estimator from Section~\ref{sec:decomposeHill} and the notation above, notice first that 
   the absolute value of  third  term~\eqref{eq:bias2} is less than
  \begin{equation}
    \label{eq:boundBias2}
    \bar b\big((1-U_{(k+1)})^{-1}\big) Z_k\,. 
  \end{equation}

Regarding the second term~\eqref{eq:bias1}, note that  for $t\le x\le y$, it holds that
$$|a(x)-a(y)|~\le~|a(x) - A| + |a(y) - A|~\le~2 \bar a(t)\,, $$
so that  the absolute value of~\eqref{eq:bias1} is less than $ 2 \bar a((1-U_{(k+1)}^{-1})$. 
We obtain the following key error decomposition. 
\begin{lemma}[An almost sure error decomposition of  the Hill estimator]\label{lem:decomposHill}
  Let $\bar F$ satisfy~\eqref{eq:rv-definition}, and let $\bar a$ and $\bar b$ denote the functions defined in~\eqref{eq:barabarb} relative to the Karamata representation of the quantile function $Q$, and let $\ZK$ be defined in~\eqref{eq:defZk}.   
  The error of the Hill estimator for fixed $k\le n$  satisfies almost surely
  \begin{align}
|\HillEst(k) - \gamaOrc|~\le~ 
    &  
      \;  \gamaOrc \; |\ZK-1 | +   2 \bar a\big((1-U_{(k+1)})^{-1}\big) + \bar b\big((1-U_{(k+1)})^{-1}\big)  \ZK\,. 
     \quad  \label{eq:deviationbound-hill} 
  \end{align}
  \end{lemma}
  The first term in the right-hand side of~\eqref{eq:deviationbound-hill} can be explicitly controlled with high probability because its distribution is known. It can be viewed as a `variance' term. The second and third terms can be viewed as bias terms, as shown next.

  \subsection{Upper bounds in probability}
  As an intermediate step, we state below a standard result regarding
  concentration of order statistics. This key result is also used in
  the working paper~\cite{clemencon2025weak}. We provide the proof for
  completeness in the Appendix,
  Section~\ref{sec:prooflem:boundUnifOrderStat}. 
\begin{lemma}[Concentration of  $U_{(k+1)}$]\label{lem:boundUnifOrderStat}
  With probability at least at least $1-\delta$,
\begin{equation}
  \label{eq:lowerUk}
  1-U_{(k+1)}~\le~\frac{ k+1}{n}\bigl(1 +R(k+1,\delta)\bigr),
\end{equation}  
with $R(k,\delta) = \sqrt{{3 \log(1/\delta)}/{k} }+
{3\log(1/\delta)}/{k}$. 
\end{lemma}
To treat the second term of~\eqref{eq:deviationbound-hill} as a bias
term, we need to bound $ Z_k $ from above with high
probability. This can be done using a quantile of the Gamma
distribution. However, we provide explicit bounds for
readability. Note that the following bound may be useful for chaining
(Remark~\ref{rem:chaining}).
\begin{lemma}[Gamma upper bound]\label{lem:tailBound_abs_Zk_1}
  With probability at least $1-\delta$, 
  \begin{equation}
    \label{eq:tailbound-Zk-delta}
    |Z_k-1|~\le~ \sqrt{\frac{2\log(2/\delta)}{k}} + \frac{\log(2/\delta)}{k}~:=~\tilde V(k,\delta) .
  \end{equation}
\end{lemma}
\begin{proof}
  The (recentered) Gamma distribution $\mathcal{G}(\alpha,\beta)$
  belongs to the class of sub-gamma distributions $\Gamma(v,c)$ with
  $v=\alpha/\beta^2$ and $c=1/\beta$ (see
  e.g.~\citet[Chapter~2]{BLM2013}. 
  Thus $ Z_k -1 \in\Gamma( v= 1/k, c = 1/k)$.  Recall also
  (see~\citet[Page~29]{BLM2013} for a proof) that a random variable
  $Z\in \Gamma(v,c)$ satisfies the tail bounds
\begin{equation}
  \label{eq:tailboundGamma}
  \PP[Z - \EE(Z)~>~ct  + \sqrt{2vt} ] \vee
  \PP[Z - \EE(Z)~<~- ct  - \sqrt{2vt} ]~\le~e^{-t}.
\end{equation}
With $\EE[Z_k] =1$, $v=1/k, c=1/k$, we
obtain 
\begin{equation*}
  \PP[Z_k - 1~>~t/k  + \sqrt{2t/k} ] \vee
    \PP[Z_k - 1~<~- t/k  - \sqrt{2t/k} ]~\le~e^{-t}. 
  \end{equation*}
  The result follows by a union bound. 
\end{proof}
We are ready to state the main result of this section. 
\begin{theorem}[Bias-Variance probability upper bound on the Hill error]\label{thm:main}
  Let $\bar F$ satisfy~\eqref{eq:rv-definition}. 
  For $1\le k\le n-1$, 
  with probability at least $1-\delta$,  the error of the Hill estimator $\HillEst(k)$ satifies
\begin{align*}
  |\HillEst(\indexk)  - \gamaOrc |
  &~\le~\gamaOrc \, V (\indexk,\quantile) +  \biasF{k,n,\delta} 
\end{align*}
where $V(k,\quantile) = V_1(k,\delta/2)$ is the  $1 - \quantile/2$-quantile of $|\Zk-1|$ and  satisfies
\begin{equation}
  \label{eq:expressionV}
 V(k,\quantile)~\le~\tilde V(k,\delta/2) ~=~ \sqrt{\frac{2\log(4/\delta)}{k}} + \frac{\log(4/\delta)}{k}\,,   
\end{equation}
and
\begin{equation}
  \label{eq:expressionB}
  \biasF{ k,n,\delta }~=~2 \bar a \Bigl( \frac{n}{(k+1)\bigl(1+ R(1,\delta/2)\bigr) }\Bigr)+
  \bigl(1 + \tilde V(1,\delta/2)\bigr) \bar b\Bigl(\frac{n}{(k+1)(1 + R(1,\delta/2 )} \Bigr)\, .  
\end{equation}
where $R$ is as defined in Lemma~\ref{lem:decomposHill}. Thus the Hill estimator satisfies
Condition~\ref{cond:genericErrorDecomp} with variance and bias
functions $V,B$ given by~\eqref{eq:expressionV}
and~\eqref{eq:expressionB}.
\end{theorem}
 The proof
follows straightforwardly from combining the bounds established in
Lemmas~\ref{lem:decomposHill},~\ref{lem:boundUnifOrderStat}, and
\ref{lem:tailBound_abs_Zk_1}. The details  are deferred to
Section~\ref{sec:proof_main_thm_hill} in the Appendix.

\begin{remark}[Explicit expression for the smallest admissible extreme sample size $k_0(\delta)$]\label{rem:expression_ko}
  The upper bound~\eqref{eq:expressionV} on the deviation term $V(k,\delta)$ chosen as the $1-\delta/2$ quantile of $|Z_k-1|$ as in Theorem~\ref{thm:main}, yields  a control of  the quantity $k_0(\delta)$ in~\eqref{eq:defk0},
$$
k_0(\delta )~\le~2^2(1+\sqrt 2)^2 \log(4 |\grid|/\delta)~\le~36 \log(4|\grid|/\delta)\,. 
$$
\end{remark}

The main conclusion of this work, stated below, is an immediate corollary from Theorem~\ref{theo:AVguarantee} and Theorem~\ref{thm:main}. 
\begin{corollary}[Main result: Guarantees of the EAV Hill estimator]\label{cor:hatk}
  Let $\bar F$ satisfy~\eqref{eq:rv-definition}, let \( V(k,\delta) \)
  be the $1 - \delta/2$ quantile of $|Z_k-1|$ where $Z_k$ follows a
  $\mathcal{G}(k,k)$ distribution (see~\eqref{eq:identityExposum}), 
  and let \( \kAV \) be defined according to the latter deviation function $V$,
  through the adaptive rule~\eqref{eq:khat-gamma}, with a triplet
  $(\deltaG,n,\grid)$ satisfying Condition~\ref{cond:wideGrid} and
  $n\ge n_0(\delta)$ defined in~\eqref{eq:defn0}.

  On a favorable event \( \mathcal{E} \) with probability at least \( 1-\delta \), the absolute error \( |\HillEst(\kAV) - \gamma| \) of the adaptive Hill estimator is less than the upper bound stated in Theorem~\ref{theo:AVguarantee}.
 \end{corollary}



%% file: Contents2/SecondOrder.tex
\section{Additional guarantees under second order conditions}\label{sec:secondorder}

Our goal in this section is to derive explicit upper bounds on the
oracle \( \kS(\delta, n) \) and its adaptive version \( \kAV \),
thereby reinforcing the guarantees established in
Propositions~\ref{prop:errorOracle}, \ref{prop:optimality_kS}, and
Corollary~\ref{cor:hatk}. By considering cases where an additional von
Mises condition holds (Condition~\ref{cond:vonMises} below), we
provide a sharper characterization of the oracle's behavior. This
analysis also facilitates direct comparisons with existing results in
the literature under comparable conditions. Our main result
demonstrates that the oracle indeed attains the optimal rate, and that
for `well chosen' grids (in a sense that shall be made precise
shortly) of size $O(\log n)$ spanning the full interval
$[0,n]$, the adaptive estimator also achieves the known minimax rate
for adaptive estimators, up to a moderate additional factor
\((\log\log n)^{\frac{1}{2(1+2|\rho|)}}\), where \(\rho\) is a
second-order parameter defined below. This further validates its
optimality in extreme value estimation.

\subsection{Second order and von Mises conditions, associated exisitng results}\label{sec:reviewSecondOrder}
In the setting of Theorem~\ref{thm:main}, we
follow~\cite{boucheron2015tail} in assuming that the following
condition holds true in addition to the regular variation
condition~\eqref{eq:rv-definition}. We recall that the latter
condition implies already that the quantile function below writes
$Q(t) = t^{\gamma}L(t)$ where $L(t)$ satisfies~\eqref{eq:karamata}.
\begin{condition}[Von Mises condition]\label{cond:vonMises}
 The tail quantile function $Q(t) = F^{\leftarrow}(1 - 1/t)$ has
  Karamata representation
  \begin{equation}\label{eq:karamataStrong}
Q(t) ~=~ A t^{\gamma}\exp \int_{t_0}^t \frac{b(u)}{u} \, \ud u\,, \qquad t\ge t_0,      
  \end{equation}
where $A>0$, $t_0>0$,  and where the function $b$ satisfies
$$\bar b(t)~:=~\sup_{x\ge t}|b(x)| \le Ct^{\rho} $$ for some $\rho<0$ and $C>0$. 
\end{condition}
Compared with the minimal regular variation assumption on $\bar F$,
the additional assumption encapsulated in
Condition~\ref{cond:vonMises} is that the function $a(t)$ relative to
the quantile function $Q(t)$ in~\eqref{eq:karamata} is constant, so
that $\bar a \equiv 0$, and that the von Mises function $b(t)$
vanishes fast enough.  As emphasized in~\cite{boucheron2015tail},
Condition~\ref{cond:vonMises} is weaker than the popular second order
`Hall condition'
(\cite{hall1985adaptive,drees1998selecting,csorgo1985kernel}), which
stipulates that
\begin{equation}
  \label{eq:secondorderClassical}
Q(t)~=~A^\gamma t^\gamma\bigl(1 + \gamma D t^\rho + o(t^\rho) \bigr). 
\end{equation}
Indeed the Hall condition~\eqref{eq:secondorderClassical} implies that the
function $b$ is regularly varying, which facilitates adaptive
estimation through estimation of the second order parameter $\rho$. Note that \citep{csorgo1985kernel} the condition in~\eqref{eq:secondorderClassical} is equivalent to 
$$
1- F(x) = B_1x^{-1/\gamma}(1+B_2 x^{\rho/\gamma} + o(x^{\rho/\gamma})), 
$$
for some constants $B_1,B_2$
We
refer to~\cite{segers2002abelian} for a thorough discussion of
second-order conditions for tail-index estimations, equivalent
statements in terms of the survival function, and their consequences
on the bias of the Hill estimator. As noted
in~\cite{boucheron2015tail}, a consequence of Theorems 1 and 2
of~\cite{segers2002abelian} is the following:
Condition~\ref{cond:vonMises} is stronger than (implies) the second
order Pareto assumption made \emph{e.g.}
in~\cite{carpentier2015adaptive}, which is that for some constants
$C>0,D$,
\begin{equation}
  \label{eq:carpentierK_assum}
  |\bar F(x) - Cx^{-1/\gamma}|~\le~D x^{(\rho-1)/\gamma}. 
\end{equation}
We now provide a minimal overview  of existing results on optimal selection
of the number $k=k(n)$ (or $k(n)+1$) of upper order statistics for
tail-index estimation, under second order
assumptions~\eqref{eq:secondorderClassical},~\eqref{eq:carpentierK_assum},
or Condition~\ref{cond:vonMises}.  Our aim is simply to shed light on
the quality of the guarantees obtained below regarding the oracle
\( \kS \) defined in~\eqref{eq:kstar} and its adaptive version
\( \kAV \) in~\eqref{eq:khat-gamma} under additional second-order
conditions. We do not intend to provide an exhaustive account of the
mathematical statistics on this subject. For a thorough review, we
refer the reader to \cite{carpentier2015adaptive} or
\cite{boucheron2015tail}.

It is shown in~\cite{hall1982some} that
if~\eqref{eq:secondorderClassical} is satisfied, then any sequence
$k^{\mathrm{opt}}(n)$ such that
$k^{\mathrm{opt}}(n)\sim \lambda n^{-2\rho/( 1 -2\rho)}$ is
asymptotically optimal in terms of asymptotic mean squared error of
the Hill estimator, and the associated asymptotic rate of convergence
is $n^{\rho/(1-2\rho)}$, meaning that
$n^{-\rho/(1-2\rho)}(\hat \gamma(k^{\mathrm{opt}}(n)) - \gamma)$
converges in distribution to a non degenerate limit.  A lower bound
derived in~\cite{hall1984best} shows that the latter rate is minimax
optimal among all possible estimators of the tail index, over a class
of functions satisfying~\eqref{eq:carpentierK_assum}.  An adaptive
procedure is next proposed in~\cite{hall1985adaptive}, based on
estimating $\rho$ and plugging the estimate $\hat \rho$ into the
expression of $k^{\mathrm{opt}}$. It is shown that under the Hall
condition~\eqref{eq:secondorderClassical}, the resulting adaptive rule
$\hat k$ is asymptotically equivalent to $k^{\mathrm{opt}}$ and thus
enjoys the same minimax error rate. A wealth of refinements of this
adaptive method have been obtained under additional conditions such as
third order ones (see \emph{e.g.} \cite{ivette2008tail} ).  An
alternative selection rule based on a Lepski-type procedure is
proposed in \cite{drees1998selecting}, offering asymptotic optimality
guarantees that similarly rely on~\eqref{eq:secondorderClassical} as
well as the first display of Condition~\ref{cond:vonMises}. Notably,
this approach does not require prior knowledge of the possible range
of the value of $\rho$ thereby somewhat relaxing the conditions
required in \cite{hall1985adaptive}.

Significant breakthroughs towards relaxing the second order
condition~\eqref{eq:secondorderClassical} have been achieved by
\cite{carpentier2015adaptive} and~\cite{boucheron2015tail}, who work
instead respectively in the relaxed
setting~\eqref{eq:carpentierK_assum} or under
Condition~\ref{cond:vonMises}. Both analyses are conducted in a non
asymptotic setting. A lower bound on the minimax error of any adaptive
estimator of the tail index is obtained in both references, which is
of order $((\log\log n)/n)^{|\rho|/(1+2|\rho|)}$ based on the
construction of an ill behaved distribution satisfying however 
Condition~\ref{cond:vonMises}. These
results confirm that adaptivity has a price, namely the error of any adaptive
estimator suffers from an  multiplicative factor
$(\log\log n)^{|\rho|/(1+2|\rho|)}$ compared with the optimal achievable rate when $\rho$ is known. 
Regarding positive results, \cite{carpentier2015adaptive}
focus on a specific tail-index estimator different from the Hill
estimator also involving the top $k$ order statistics, and they
propose an adaptive selection rule for $k$, for which they prove upper
bounds on the error with rates matching the lower bound.  On the other
hand, \cite{boucheron2015tail} consider a Lepski-type rule similar in
spirit to~\cite{drees1998selecting}'s preliminary selection rule  for the Hill estimator, and they
derives guarantees for an adaptive estimator (Equation~3.7 of the cited reference) with a stopping
criterion involving  a sequence of tolerance thresholds $r_n(\delta)$. The latter incorporate quantities that are arguably difficult to track. The authors recognize that their guarantees intend to
serve as reassuring guidelines, and they investigate in their
simulation study a different, simplified rule where all unknown
constants and intermediate sequences are replaced with a term
$2.1\, \log\log n$.  In contrast, our stopping criterion is based
solely on explicit constants from the very start, and our guarantees
directly apply to the rule implemented in our simulation study.

We now analyze the error of the Hill estimator using
\( k = \kS(\delta, n) \) upper order statistics. We proceed in two
steps: In Section~\ref{sec:lowerBoundKstar} we establish a lower
bound on $\kS$ under Condition~\ref{cond:vonMises},  
thereby obtaining an explicit expression for the minimum sample size
$n_0(\delta)$ introduced in~\eqref{eq:defn0}. In Section~\ref{sec:OracleErrorBound} we leverage this
lower bound together with previously proved upper bounds on the oracle
error involving $\kS$, to establish an explicit bound on the oracle
error. Finally we use Theorem~\ref{theo:AVguarantee} to control the error of $\HillEst(\kAV)$ based on the  oracle error, \ie the error of $\HillEst(\kS)$. 

\subsection{A lower bound on $\kS$ and an explicit expression for $n_0(\delta)$}\label{sec:lowerBoundKstar}
Under Condition~\ref{cond:vonMises}, the bias function $B$  in the error bound stated in Theorem~\ref{thm:main}    satisfies 
\begin{equation}\label{eq:bound-bias-vonMises}
  B(k,n,\delta)~\le~C_1(\delta,\rho) \Big(\frac{n}{k+1}\Big)^{\rho}, 
\end{equation}
where
\begin{align}
  C_1(\delta,\rho )&~=~C\bigl(1+\tilde V(1,\delta/2)\bigr)\bigl(1+R(1,\delta/2)\bigr)^{-\rho} \nonumber \\
                 &~=~C\bigl(1+ \sqrt{ 2\log(4 /\delta)}  +
                   \log(4/\delta)\bigr) 
                   \bigl( 1 + \sqrt{3\log(2/\delta)} +
                   3\log(2/\delta)\bigr)^{-\rho} \nonumber \\
  &~\le~C\big( 1 + \sqrt{3\log(4/\delta)} +
                   3\log(4/\delta)\big)^{1-\rho} . \label{eq:boundConstantDelta}
\end{align}

On the other hand,  the deviation bound $V(k,\delta)$ is provably bounded from below by a subgamma-type quantile, as shown next. 
\begin{lemma}[Gamma's absolute deviations: lower bounds on the quantiles]\label{lem:gamma_quantile_lower}
  Let $V(k,\delta)$ denote the $1-\delta/2$ quantile of the $|Z_k-1|$,  as in Eq.~\eqref{eq:expressionV} from the statement of Theorem~\ref{thm:main}, where we recall $Z_k\sim\mathrm{Gamma}(\alpha=k,\beta=k)$. There exists universal constants, $0<c_1\le 1$ and $0<c_2\le 2$  such that for all $(k,\delta)$,
  $$
  V(k,\delta)~\ge~c_1\Big(  \sqrt{\frac{0\vee\log(c_2/\delta)}{k}} + \frac{\log(c_2/\delta)}{k}\Big) .
  $$
\end{lemma}
The proof is deferred to Appendix~\ref{sce:proof_gamma_quantile_lower}.
Combining the upper bound~\eqref{eq:bound-bias-vonMises} on
$B(k,n,\delta)$ offered by Condition~\ref{cond:vonMises} together with
the lower bound on $V(k,\delta)$ from
Lemma~\ref{lem:gamma_quantile_lower} we are ready to state a  lower
bound on $\kS(\delta,n)$. Assumptions on the grid are necessary, stated below, in addition to Condition~\ref{cond:wideGrid}. 
\begin{condition}[Fine enough grid]\label{cond:grid2ndOrder}
The grid \(\grid = \{k_m, m \leq |\grid|\}\) is chosen such that:
\[
\frac{k_{m+1}}{k_m} \leq \beta \quad \text{for some} \quad \beta > 1 \quad \text{and for all} \quad m \leq |\grid|.
\]
\end{condition}

\begin{remark}[Grid choice]\label{rem:choiceOfGrid}
  For moderately large $n$ and a fixed $\delta \in (0,1)$, there are
  several natural choices for $\grid$ ensuring that
  both Condition~\ref{cond:wideGrid} and~\ref{cond:grid2ndOrder} are satisfied. Specifically, if
  $\biasF{1, n, \delta} < \gamma \varianceF{ 1 ,\delta}$ and
  $\biasF{n,n,\delta}> \gamma \varianceF{n,\delta}$, then the
  exhaustive grid $\grid =\{1,\ldots, n\}$ and geometric grids
  $\grid= \{\lfloor \beta^m \rfloor: 0\le m\le \log_\beta(n) \}$ for
  some $\beta>1$, satisfy Conditions~\ref{cond:wideGrid} and~\ref{cond:grid2ndOrder}. Additionally,
  for sufficiently large $M<n$ such that
  $\biasF{n/M, n, \delta} < \gamma \varianceF{n/M,\delta}$, the
  uniform grid $\grid = \{ \lfloor m n/M \rfloor, 1\le m\le M\}$ also
  satisfies Condition~\ref{cond:wideGrid} and Condition~\ref{cond:grid2ndOrder} with $\beta=2$. 
\end{remark}
\begin{proposition}[A lower bound on $\kS$ under Conditions~\ref{cond:wideGrid},\ref{cond:vonMises} and~\ref{cond:grid2ndOrder}]
   \label{prop:lowerBoundKstar}
   Let \((\delta,n,\grid)\) satisfy Conditions~\ref{cond:wideGrid}
   and~\ref{cond:grid2ndOrder},  and let the data distribution satisfy
   Condition~\ref{cond:vonMises}. For any \(\delta\) such that 
   \begin{equation}\label{eq:condition_delta_for_lowerBound_C2}
     0 < \delta \leq  c_2^2/4, 
   \end{equation}
   where \(c_2\) is defined in
   Proposition~\ref{lem:gamma_quantile_lower}, the oracle \(\kS\)
   satisfies
\[
  \kS(\delta, n) \geq \beta^{-1} \left( \frac{C_2(\rho)
      \gamma^{2/(1-2\rho)}}{\log(4/\delta)} n^{-2\rho/(1-2\rho)} - 1
  \right),
\]
for \(n\) large enough so that \(n/2\) is at least as large as the
lower bound in the above expression. Here,
\[
  C_2(\rho) = \left(\frac{4}{21}\right)^2
  \left(\frac{c_1}{\sqrt{2}C}\right)^{2/(1-2\rho)},
\]
where \(c_1\) is as defined in Lemma~\ref{lem:gamma_quantile_lower},
and \(C\) is as specified in Condition~\ref{cond:vonMises}.
\end{proposition}
The proof is  deferred to Appendix~\ref{sec:proof_prop_lowerBoundKstar}. 
Remark~\ref{rem:expression_ko} combined with the lower bound on $\kS$ in Proposition~\ref{prop:lowerBoundKstar} yields immediately the following control of $n_0(\delta)$. 
\begin{corollary}[Control of minimum required sample sizes]\label{cor:minsamplesize}
  Let $(\deltaG,n,\grid)$ satisfy Conditions~\ref{cond:wideGrid}, \ref{cond:grid2ndOrder}. 
  Under Condition~\ref{cond:vonMises}, the minimum sample size
  $n_0(\delta)$ defined in~\eqref{eq:defn0} with $V(k,\delta)$ as in
  Theorem~\ref{thm:main}, is no greater than the first integer $n$
  such that 
$$
36\log(4 |\grid|/\delta)~\le~ \beta^{-1} \left(\frac{ C_2(\rho) ~ 
\gamma^{2/(1-2\rho)}  }{\log(4|\grid|/\delta)} ~~  n ^{ - 2 \rho /(1 - 2 \rho)}\, - 1\right)\,.
$$
\end{corollary}

\subsection{Oracle and adaptive error bounds}\label{sec:OracleErrorBound}

The following error bound is a consequence of the previously
established lower bound of $\kS$
(Proposition~\ref{prop:lowerBoundKstar}), of the control of the oracle
error \emph{via} the deviation term alone
(Proposition~\ref{prop:errorOracle}) and the upper
bound~\eqref{eq:expressionV} on the deviation term in
Theorem~\ref{thm:main}.  It shows that under
Condition~\ref{cond:vonMises}, for grid choices such that
Conditions~\ref{cond:wideGrid} and~\ref{cond:grid2ndOrder} hold, the
oracle $\kS(\delta,n)$ indeed meets the minimax rate of convergence
for the Hill estimator obtained in~\cite{hall1982some} when
considering a restricted class of distributions satisfying the
(stronger) second order
assumption~\eqref{eq:secondorderClassical}. This brings further
justification for considering $\kS( \delta, n )$ as an oracle rule, in
addition to the optimality properties
(Propositions~\ref{prop:errorOracle},~\ref{prop:optimality_kS}) stated
in Section~\ref{sec:AVframework}, which are valid under weaker
assumptions. 

\begin{theorem}[Error bound for the oracle]\label{theo:errorOracle_vonMises}
  Let $(\delta,n,\grid)$ satisfy Conditions~\ref{cond:wideGrid} and
  Condition~\ref{cond:grid2ndOrder} for some $\beta>1$, and let the
  data distribution satisfy Condition~\ref{cond:vonMises}. Assume in
  addition that $\delta$
  satisfies~\eqref{eq:condition_delta_for_lowerBound_C2}, and that
  $\kS(\delta,n)\ge \log(4/\delta)$.  Then, the error for the
  oracle~$\kS(\delta,n)$ satisfies, with probability at least
  $1-\delta$,
  \begin{equation*}
    \bigl|\hat \gamma\bigl(\kS(\delta,n)\bigr) - \gamma\bigr|~\le~\frac{2 (1+\sqrt{2})\sqrt{\beta} }{\sqrt{ C_2(\rho)}}\sqrt{1+\log(4/\delta)} \;
 (n/\gamma^2)^{\frac{\rho}{1-2\rho}},
\end{equation*}
where  $ C_2(\rho)$ is as in   Proposition~\ref{prop:lowerBoundKstar}.  
\end{theorem}
The proof is  deferred to  Appendix~\ref{sec:prooferrorOracle_vonMises}. 
The proof reveals that, under the assumptions of the statement,
the deviation term $V(\kS(\delta,n), \delta)$ satisfies
\begin{equation*}
  V(\kS(\delta,n), \delta)~\le~  \frac{(1+\sqrt{2}) ~\sqrt{\beta}}{ C_2(\rho)^{-1/2}}\;   \sqrt{1+\log(4/\delta)}\; \; \gamma^{\frac{- 1}{1-2\,\rho}} \; n ^{\frac{ \rho }{1- 2\, \rho}}. 
\end{equation*}

Combining the above bound with Theorem~\ref{theo:AVguarantee} yields
immediately a tail bound on the error of the adaptive rule $\kAV$
stated in Corollary~\ref{cor:bound_error_kav_secondOrder} below. Note
that from Proposition~\ref{prop:lowerBoundKstar}, the condition that
$\kS\ge \log(4|\grid|/\delta)$ in the statement below is satisfied for
sample sizes large enough such that
$ C_2(\rho) ~ \gamma^{1/(1-2\rho)}~~ \frac{n ^{ - 2 \rho /(1 - 2
    \rho)}}{\log(4|\grid|/\delta)} - 1 \ge \log(4|\grid|/\delta)$.
 
\begin{corollary}[Error bound for $\kAV$]\label{cor:bound_error_kav_secondOrder}
  Let Condition~\ref{cond:vonMises} be satisfied, and let
  $\delta,n,\grid$ be such that $(\deltaG, n, \grid)$ satisfy
  Conditions~\ref{cond:wideGrid} and~\ref{cond:grid2ndOrder}.
  Assume in addition that $\deltaG$
  satisfies~\eqref{eq:condition_delta_for_lowerBound_C2}, and that 
  $\kS(\deltaG, n)\ge \log(4|\grid|/\delta)$.

  Then the
  error for the adaptive validation Hill estimator
  $\hat\gamma(\kAV)$ 
  satisfies, with probability at least $1-\delta$,
  \begin{equation*}
    |\hat \gamma(\kAV) - \gamma|~\le~ \frac{6\gamma V^*}{1 - 6 V^*}, 
\end{equation*}
where $V^* = V(\kS(\delta/|\grid|), \delta/|\grid|)$ satisfies
$$
V^*~\le~\frac{(1+\sqrt{2}) ~\sqrt{\beta}}{ C_2(\rho)^{-1/2}}\;   \sqrt{1+\log(4|\grid|/\delta)}\; \; \gamma^{\frac{- 1}{1-2\,\rho}} \; n ^{\frac{ \rho }{1- 2\, \rho}}. 
$$
\ 
\end{corollary}

\begin{remark}[Minimax optimality, grids of logarithmic size, and chaining.]
  \label{rem:chaining}
  For \(|\grid| \leq D \allowbreak\log n\) with some \(D > 0\), \(\kAV\) incurs an adaptivity cost, i.e., an additional multiplicative factor compared to the oracle error \(\kS(\delta, n)\) (see Theorem~\ref{theo:errorOracle_vonMises}), of order \(\sqrt{\log \log(n)}\). This is the case in particular for geometric grids mentioned earlier in Remark~\ref{rem:choiceOfGrid}, which also satisfy Conditions~\ref{cond:wideGrid} and~\ref{cond:grid2ndOrder}. The full grid \(\grid = \{1, \ldots, n\}\) has a multiplicative factor of \(\log(n)\). Uniform grids \(k_m = \lfloor mn/M \rfloor\), \(1 \leq m \leq M\) with \(M = D \log n\), also fit within this framework under the condition that \(n/(D \log n) < \kS(\deltaG, n)\). However, the latter condition is not guaranteed for small values of \(|\rho|\) in view of the lower bound on \(\kS\) in Proposition~\ref{prop:lowerBoundKstar}.

From a theoretical viewpoint, we recommend the use of a geometric grid. However, our experiments suggest that uniformly spaced grids have comparable performance. An informal explanation for this good behavior is as follows: inspection of the proof of Theorem~\ref{theo:errorOracle_vonMises} reveals that the only necessary condition on the grid is that \(k_{m^*+1}/k_{m^*}\) is not too large, where \(m^*\) is the grid index such that \(\kS = k_{m^*}\). Therefore, the structure of the grid matters only in the neighborhood of \(\kS\). For large enough \(n\), \(\kS\) is sufficiently large so that even with a uniformly spaced grid, \(k_{m^*+1}/k_{m^*}\) is small enough to result in good performance.

In comparison, the adaptive methods proposed in~\cite{boucheron2015tail} and~\cite{carpentier2015adaptive} pay the price of adaptivity with an additional multiplicative factor of order \((\log \log(n))^{|\rho|/(1 + 2|\rho|)}\). Both cited references (and in particular~\cite{boucheron2015tail} under Assumption~\ref{cond:vonMises}) show minimax optimality of this multiplicative factor among adaptive estimators. One can conclude that the adaptive validation method seems to benefit from a nearly free lunch: if the second-order conditions are not met, it still enjoys the guarantees stated in Proposition~\ref{prop:errorOracle} and Proposition~\ref{prop:optimality_kS}. Moreover, if the second-order condition~\ref{cond:grid2ndOrder} is satisfied, it attains a nearly minimax adaptive rate, matching the lower bound established in~\cite{carpentier2015adaptive} and~\cite{boucheron2015tail} up to an additional multiplicative factor
\[
\frac{\sqrt{\log \log n}}{(\log \log(n))^{|\rho|/(1 + 2|\rho|)}} = (\log \log n)^{\frac{1}{2(1 + 2|\rho|)}},
\]
which is nearly constant, especially for large values of \(|\rho|\).

Finally, as mentioned earlier, chaining techniques similar to those used, for example, in~\cite{boucheron2015tail} yield a finer uniform control of the fluctuations (\(|\hat{\gamma}(k) - \gamma|\), \(k \leq \kS(\delta, n)\)) at the price of larger constants in the deviation term \(V(k, \delta)\). These techniques allow replacing the quantities \(\delta/|\grid|\) by \(\delta/\log(n)\) to obtain a \(\sqrt{\log \log n}\) multiplicative factor even with a grid of size \(n\). However, preliminary experiments showed no improvement over the basic union bound approach taken in Lemma~\ref{lem:UnifDev-eventE} on a grid of logarithmic size, precisely because of the larger constants involved, leading to late stopping.

In fact, the 'geometric grid' approach may be understood as a concrete implementation of the theoretical chaining technique, leading to similar results and simplifying the computational costs. This approach also allows keeping explicit deviation controls in terms of the quantiles of the centered Gamma variable quantity \(|Z_k - 1|\) (see~\eqref{eq:identityExposum}), which can be easily calibrated.

\end{remark}

%% file: Contents2/Simulations.tex
\makeatletter
\newcommand*{\deq}{\mathrel{\rlap{%
      \raisebox{0.3ex}{$\m@th\cdot$}}%
    \raisebox{-0.3ex}{$\m@th\cdot$}}=}
\makeatother
\noindent The code for our experiments is available at \url{https://github.com/mahsa-taheri/EAV}.

\noindent
\textbf{Compared methods and general objectives}
Our experiments aim to compare the performance of the proposed Extreme Adaptive Validation  method with two existing state-of-the-art approaches for adaptive Hill estimation under comparable assumptions, specifically those of Boucheron et al. (2015) and Drees et al. (1998). For a discussion of their respective underlying assumptions, refer to Section~\ref{sec:reviewSecondOrder}. Focusing on quantitative performance assessment, we restrict our analysis to simulated data experiments, which provide a ground truth for evaluation.
Our results demonstrate that the
proposed EAV method effectively outperforms the approaches of
\cite{boucheron2015tail} and \cite{drees1998selecting} when dealing
with very ill-behaved distributions, while showing comparable performance
for distributions with survival functions converging rapidly to their limiting power law. We confirm
the finding of \cite{boucheron2015tail} that
\cite{drees1998selecting}'s asymptotic approach outperforms adaptive rules grounded on  non-asymptotic guarantees (that is, EAV and \cite{boucheron2015tail}'s method) for particularly well-behaved distributions with rapidly converging tail behavior. 

The adaptive rules of \cite{boucheron2015tail} and \cite{drees1998selecting} are implemented with the exact same parameters and calibrated constants as the ones  proposed by~\citet[Page~30;
Equations~(5.1)
and~(5.2)]{boucheron2015tail}. 
For $\kbch$ and $\kdk$, we set the lower limit of the admissible range for $k$ to $l_n = 30$, following the recommendations and notation of \cite{boucheron2015tail}.
Regarding \cite{boucheron2015tail}'s approach, 
It is worth noting that these
constants are slightly different from those for which they provide
theoretical guarantees.

Regarding  the EAV method,
we use logarithmic grids of the form 
$\grid = \{k_m = \lfloor\beta^m\rfloor, 1 \le  m \le \lfloor \log n / \log(\beta) \rfloor\}$, so that  $|\grid| = \log n / \log(\beta)$ with $\beta = 1.1$ and $1\le m\le |\grid|$. Additional results for other grid choices, including the uniformly spaced  grid, are provided in Appendix~\ref{sec:adsim}.  
We implement  the stopping condition proposed in~\eqref{eq:khat-gamma}, that is, 
\begin{equation*}
\begin{aligned}
\kAV ~=~ \max \Bigl\{
&k \in \grid, \, k \geq k_0(\delta), \text{ such that } \forall j \in \grid \text{ with } j \leq k, \\
&\left| \HillEst(k) - \HillEst(j) \right| ~ \leq ~ \frac{\HillEst(k)}{1 - 2 \variance(k, \deltaG)} \bigl( \variance(j, \deltaG) + 3 \variance(k, \deltaG) \bigr)
\Bigr\},
\end{aligned}
\end{equation*}
where $\variance(k,\deltaG)$ is a quantile of order $1-\deltaG/2$ of
$|\Zk-1|$, with $\Zk\sim
\mathcal{G}(\alpha=k,\beta=k)$. 
Here, with a geometric grid  as described above  and
$\beta=1.1$, we have $|\grid|\approx \log n / \log \beta $. We fix
$\delta=0.9$ and we compute numerically the
$(~1 - (\delta\log \beta)/(2 \log n) ~)$-quantile of $|Z_k-1|$ by
Monte-Carlo sampling with $N=2000$ independent draws of $Z_k$.  
The confidence level (\(\delta\)) is the sole free parameter in the proposed EVA method. We recommend selecting a relatively high value (e.g., \(0.9\)) to mitigate the pessimistic upper bound on the variance term, obtained via the union bound technique over multiple candidate values of \(k\). Smaller values of \(\delta\) often result in delayed stopping and poorer performance. 

\noindent 
\textbf{Datasets}
We assess the performance of the three methods
using $\nbrsamp=10\,000$ samples
generated from the distributions listed below to estimate the unknown
parameter~\gamaOrc\ with the Hill estimator. Smaller sample sizes ($n=1\, 000$) are considered in Table~\ref{tab:SimMSEsmall}, with almost similar conclusions. 
 We consider the following data  distributions:
 \begin{enumerate}
 \item 
An ill-behaved distribution which is regularly varying but does not satisfy~the standardized Karamata representation~\eqref{eq:karamataStrong} (see Appendix~\ref{sec:counterExample} for a proof.) To our best knowledge, this counter-example is new. This distribution, parametrized by the tail index $\gamma>0$ and an additional scaling  parameter $s\in (0,1]$ is designed so that its density vanishes infinitely often on intervals of the kind $[(n^{1/s}+(n+1)^{1/s})/2, (n+1)^{1/s}]$. Namely  we consider the distribution $F_{s,\alpha}$ of a random variable $X$ defined by 
$$X = \lfloor Z^s \rfloor^{1/s}  + \frac{1}{2}(Z-\lfloor Z^s \rfloor^{1/s}),$$
where $Z$ follows a Pareto distribution with tail index $\gamma=1/\alpha$. 
  In our simulation we consider $\alpha=2, s\in\{2/3, 1/2\}$. 
 \item Symmetric $\alpha$-stable distributions with $\alpha\in \{1.5,1.7,1.99\}$ ($\alpha=1/\gamma$), denoted by $S_{1.5}$, $S_{1.7}$, and $S_{1.99}$. 
These distributions have  characteristic function \(\mathbb{E}[e^{itX}] = e^{-|t|^{\alpha}}\). It is known \citep{csorgo1985kernel} that for \(\alpha > 1\), these distributions satisfy \eqref{eq:secondorderClassical}, and thus also Condition~\ref{cond:vonMises}. Indeed, \citep[][Chapter 2, Sections 3, 4]{ibragimov1975independent} their density can be expressed as the restriction to \(\mathbb{R}\) of a complex function that is holomorphic everywhere in the complex plane. This allows for asymptotic expansions of the form:
$ f(x) = \sum_{n=0}^N a_n x^{-n\alpha-1} + o(x^{-N\alpha-1})$,  as $  x \to \infty$. However convergence is slow and the Hill estimatir is generally ill-behaved as $\alpha\to 2$

  \item A distribution with survival function of the form $\bar{F}(x)=1-F(x)=cx^{-\alpha}(\log x)^{\beta}$, where $\beta\ne 0$, $\alpha,c>0$,   with $c=(e \alpha/\beta)^{\beta}$ defined on the domain $[x_0,\infty)$ where $x_0= \exp(\beta/\alpha)$. This distribution is referred to as the `Perturb' distribution in ~\cite{Resnick2007}, p. 87. In our simulation we fix   $\alpha=2$ and $\beta=1$ and we denote this distribution as $L_{2,1}$. This distribution has standardized Karamata representation~\eqref{eq:karamataStrong} but it does not satisfy Condition~\ref{cond:vonMises}, thus it does not satisfy~\eqref{eq:secondorderClassical} either.   
  We refer to Section~\ref{sec:xt} for proofs of our claims about the properties of $L_{\gamma,\beta}$.

 \item A Fréchet distribution with $\FrechetPar=1$ and a shift of 10, denoted by $F_{1,10}$, which satisfies~\eqref{eq:secondorderClassical} with $\gamma=1,\rho=-1$, as easily seen by an expansion of the survival function at infinity. 
 \item A Fréchet distributions with $\FrechetPar=1$, denoted by  $F_{1}$. Again $F_1$ satisfies~\eqref{eq:secondorderClassical} with $\gamma=1$, $\rho=-1$, although the convergence of the tail towards a power law is faster than for the shifted version. 

 \item A Pareto Change Point distribution (PCP) defined by 
 $$
 \overline{F}(x) = x^{-1/\gamma'} \un\{1\le x\le \tau\} + \tau^{-1/\gamma'} (x/\tau)^{-1/\gamma} \un\{ x> \tau\}, 
 $$ similar to the one considered in \cite{boucheron2015tail}. We consider three cases, $(\gamma',\gamma)=(1.0,1.1)$, $(\gamma',\gamma)=(1.0,1.25)$, and 
 $(\gamma',\gamma)=(1.0,1.5)$ with values for $\tau$ chosen such that the tail probability equal respectively $1/25$, $1/25$, and $1/15$. 
 Since the PCP tail is exactly Pareto above a finite threshold, it satisfies the strongest second order assumption \eqref{eq:secondorderClassical} considered here. 

 \end{enumerate} 
 These distributions span a broad spectrum of tail
 behaviours. Distributions in Family 1 (Counter-example) is a
 particularly ill-behaved example where no second order condition is
 satisfied, not even the standardized Karamata
 representation~\eqref{eq:karamataStrong}.  
Stable distributions in  Family 2. are a typical example where good asymptotic properties do not prevent poor finite-sample behavior. When \(\alpha \to 2\), the Hill estimator of \(\gamma = 1/\alpha\) is notoriously ill-behaved because convergence in \eqref{eq:secondorderClassical} is slow \citep{resnick2007heavy,nolan2020univariate}.  %
Distribution 3 (Perturb)
 satisfies~\eqref{eq:karamataStrong} but with von Mises function
 $b(t)$ that decreases too slowly for Condition~\ref{cond:vonMises} to
 holds. It is well known to produce `Hill horror plots'. Another
 example of `horror plot' distributions is provided by the symmetric
 stable distributions (Family 2), although they satisfy the strongest
 second order assumption considered here,
 namely~\eqref{eq:secondorderClassical}.  Intermediate cases are
 provided by shifted distributions (Family 4), which satisfy
 also~\eqref{eq:secondorderClassical} with a faster convergence to the
 limiting power law behaviour, although they are known to produce Hill
 horror plots as well.  The Fréchet distribution (Family 6) is a
 typical example of a well-behaved distribution where Hill estimation
 is easy and $k$ is easily chosen either by a visual inspection of
 Hill plots or by automatic adaptive procedures. Finally, Family 7 may be viewed as a counterexample to the relevance of second-order assumptions for practical (i.e., finite sample) purposes. Although it satisfies strong second-order assumptions, it results in poor behavior of the Hill estimator when the cut-off point is unknown. In \cite{boucheron2015tail}, it is presented as a typical example where asymptotic strategies, such as those in \cite{drees1998selecting}, fail in comparison to a non-asymptotic approach.

\noindent
\textbf{Performance metric}
As a performance metric for a fixed $k$ we consider the standardized mean squared error (MSE)  of the Hill-estimator
$ 
	\textrm{MSE}(k) = 
    \EE\Bigl[\Bigl(\frac{\HillEst(\indexk)}{\FrechetPar}-1\Bigr)^2\Bigr]\,,
$ 
where 
$\HillEst(k)$ is the Hill estimator for $k$ extreme order statistics. 
As an  illustration, Figure~\ref{fig:MontRMSE} displays  the standardized root mean squared error RMSE$(k) = \sqrt{\textrm{MSE}(k)}$ for  datasets of sample size  $n= 10\,000$ generated from  one representant of the 6 distribution families described above, as a function of $k$.   The expectation in the definition of MSE is approximated by Monte-Carlo sampling with $N=500$ replications.  \\~

\textbf{Results}
Table~\ref{tab:SimMSE} reports the estimated MSE and the standard error of this estimator (both multiplied by $100$) associated with the different adaptive rules considered here, namely EAV, \cite{boucheron2015tail}'s method and \cite{drees1998selecting}'s, denoted respectively by $\kAV$, $\kbch$ and $\kdk$, over $N=500$ experiments. Namely the quantity reported for $\hat k \in \{\kAV, \kbch,\kdk\}$ is 
$
 \widehat{\textrm{MSE}}(\hat k)~=~ N^{-1}\sum_{i=1}^N  ( \hat{\gamma}_i(\hat k_i)/\gamma - 1 )^2
$, 
where $\HillEst_i(\hat k_i)$ is the Hill estimator obtained at the $i^{th}$ replication, using the number of extreme order statistics $\hat k_i$ which is the output  of the adaptive rule $\hat k$ applied to the $i^{th}$ generated dataset.
Additionally we provide the standard error of this estimator of the expected MSE, namely
$
\mathrm{stderr}(\hat k)~ = ~N^{-1}  \big(\sum_{i=1}^N \big[ ( \hat{\gamma}_i(\hat k_i)/\gamma - 1 )^2 - \widehat{\MSE}(\hat k) \big]^2 \big)^{1/2}$.
The first five rows of Table~\ref{tab:SimMSE} reveal that for very ill-behaved distributions, specifically those from Family~1 (Counter-example) and Family 2 (Stable), $\kAV$ outperforms both $\kbch$ and $\kdk$ in terms of standardized MSE. The 'Perturb' distribution (Family 3) stands apart, as it is relatively ill-behaved for Hill estimation; however, $\kdk$ and $\kbch$ demonstrate better performance than $\kAV$ in this case. For well-behaved  distributions  satisfying the classical second-order assumption~\ref{eq:secondorderClassical}, the asymptotic rule $\kdk$ outperforms both $\kAV$ and $\kbch$. The PCP case (Family 6.) is a counter-example where the asymptotic $\kdk$ method struggles compared with $\kAV,\kbch$, despite its tail regularity. 
 To better understand the differences in behavior between $\kAV$, $\kbch$, and $\kdk$, we present in Table~\ref{tab:Simvonmises30Q} (in the Appendix), for $\hat{k} \in \{\kAV, \kbch, \kdk\}$, the range of values for $k$ produced by these three rules, along with the average values,  denoted as $\bar{k}_{EAV}$, $\bar{k}_{BT}$, and $\bar{k}_{DK}$, respectively.  The results indicate that $\kAV$ consistently yields larger values of $k$. In contrast, the range of $k$ values output by $\kbch$ and $\kdk$ is wider, often including indices near the lower limit $l_n = 30$.

  \begin{figure}
     \centering
     \includegraphics[width=0.8\linewidth]{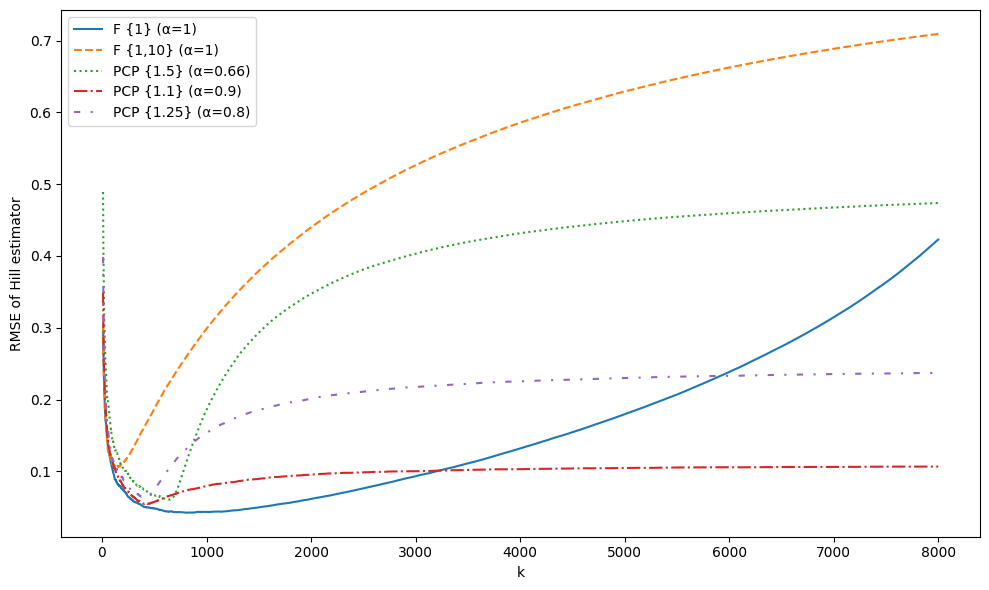}
     \includegraphics[width=0.8\linewidth]{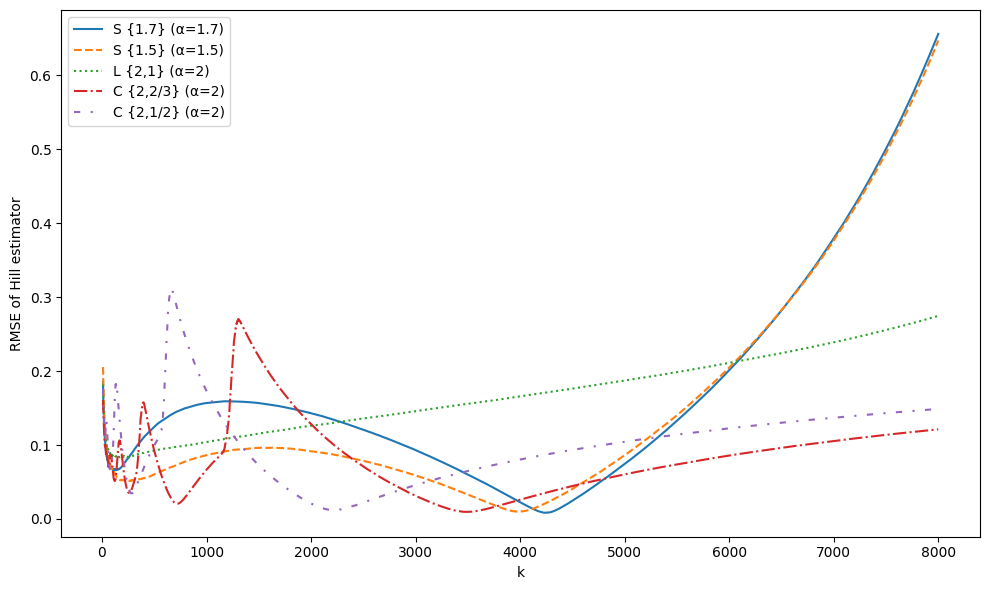}
     \caption{Monte-Carlo estimates of the standardised RMSE of Hill estimators as a function of the number of order statistics $k$ for samples of size $10\,000$ from the sampling distributions.
}
     \label{fig:MontRMSE}
 \end{figure}

\begin{table}[H]
	\centering
	\scalebox{0.8}{
		\begin{tabular}{cc|ccc}
			\hline
&  &&$n=10\,000$ & \\
			d.f. & $\gammaS$ &  $\widehat{\textrm{MSE}}(\kAV);(\operatorname{stderr}) \times 100$  &$\widehat{\textrm{MSE}}(\kbch);
            (\operatorname{stderr}) \times 100$&  $\widehat{\textrm{MSE}}(\hat{k}_{DK});
            (\operatorname{stderr}) \times 100$\\
   \hline
            $C_{2,2/3}$& 0.5     & $\textbf{01.06};(0.03)$   & $03.05;(0.16)$  & $13.80;(0.88)$ \\
			\hline
            $C_{2,1/2}$& 0.5     & $\textbf{04.55};(0.11)$   & $05.14;(0.24)$  & $29.01;(5.99)$ \\
			\hline
			$S_{1.7}$ & $1/1.7$  & $\textbf{01.28};(0.08)$  & $04.01;(0.10)$ & $04.40;(0.06)$ \\
			\hline
			$S_{1.5}$ & $1/1.5$ &$\textbf{00.35};(0.01)$  & $01.33; (0.05)$  & $01.23;(0.03)$\\
			\hline
            $S_{1.99}$ & $1/1.99$ &$\textbf{26.20};(0.03)$  & $47.70; (0.24)$  & $45.88;(0.48)$\\
			\hline
			$L_{2,1}$  & $0.5$  & $21.40;(0.14)$  & $04.47;(0.14)$ & $\textbf{04.40};(0.08)$\\
			\hline
            \hline 
			$F_{1,10}$ & $1.0$  & $12.08;(0.13)$   & $03.49;(0.09)$  & $\textbf{01.78};(0.12)$ \\
			\hline
			$F_1$& $1.0$& $06.36;(0.06)$  &  $00.69;(0.05)$ &  $\textbf{00.24};(0.03)$  \\
            \hline 
            $PCP_{1.1}$ & $1.1$ & $\textbf{00.62};(0.01)$  & $00.70;(0.11)$ & $00.73;(0.04)$ \\
			\hline
            $PCP_{1.5}$ & $1.5$ & $05.22;(0.05)$  & $\textbf{00.58};(0.07)$ & $05.99;(0.36)$ \\
			\hline
            $PCP_{1.25}$ & $1.25$ & $03.39;(0.02)$  & $\textbf{01.13};(0.06)$ & $03.34;(0.02)$ \\
            \hline
	\end{tabular}}
	\caption{
    $\widehat{\textrm{MSE}}(\hat k)$ (stderr $(\hat k)$), both multiplied by $100$, for $\hat k \in\{\kAV,\kbch,\kdk\}$, over $N=500$ experiments. 
    }

	\label{tab:SimMSE}
	
\end{table}

\begin{table}[H]
	\centering
	\scalebox{0.8}{
		\begin{tabular}{cc|ccc}
			\hline
&  &&$n=1\,000$ & \\
			d.f. & $\gammaS$ &  $\widehat{\textrm{MSE}}(\kAV);(\operatorname{stderr}) \times 100$  &$\widehat{\textrm{MSE}}(\kbch);
            (\operatorname{stderr}) \times 100$&  $\widehat{\textrm{MSE}}(\hat{k}_{DK});
            (\operatorname{stderr}) \times 100$\\
            \hline
            $C_{2,2/3}$& 0.5     & $\textbf{01.63};(0.04)$&$04.50;(0.22)$ & $07.15;(0.18)$\\
             \hline
            $C_{2,1/2}$& 0.5     &$\textbf{02.77};(0.07)$& 07.60;(0.26)& $09.15;(0.19)$ \\
			\hline
			$S_{1.7}$ & $1/1.7$  &$52.48;(0.71)$&\textbf{04.62};(0.15)& $06.73;(0.21)$ \\
            \hline
			$S_{1.5}$ & $1/1.5$  &$75.63;(0.92)$&\textbf{01.17};(0.07)& $02.35;(0.11)$ \\
			\hline
			$S_{1.99}$ & $1/1.99$ & $\textbf{24.95};(0.42)$&$37.70;(0.32)$& $106.94;(37.95)$\\
			\hline
			$L_{2,1}$  & $0.5$ & $56.39;(0.30)$&$09.28;(0.27)$& $\textbf{07.23};(0.22)$\\
			\hline
            \hline 
			$F_{1,10}$ & $1.0$  &$31.61;(0.29)$&$14.67;(0.28)$& $\textbf{06.32};(0.25)$\\
            \hline 
            $F_{1}$ & $1.0$  &$64.74;(0.35)$&$\textbf{01.92};(0.08)$& $11.44;(0.07)$\\
            \hline 
            $PCP_{1.1}$ & $1.1$ &$\textbf{00.80};(0.02)$& $00.95;(0.04)$& $00.82;(0.03)$ \\
			\hline
   $PCP_{1.5}$ & $1.5$ &$10.12;(0.07)$& $\textbf{07.71};(0.16)$& $09.57;(0.09)$ \\
			\hline
            $PCP_{1.25}$ & $1.25$ &$03.73;(0.04)$& $\textbf{03.46};(0.06)$& $03.55;(0.06)$ \\
			\hline
	\end{tabular}}
	\caption{
    $\widehat{\textrm{MSE}}(\hat k)$ (stderr $(\hat k)$), both multiplied by $100$, for $\hat k \in\{\kAV,\kbch,\kdk\}$, over $N=500$ experiments. 
    }

	\label{tab:SimMSEsmall}
	
\end{table}

%% file: Contents2/Proofs.tex
\subsection{Proof of Lemma~\ref{lem:boundUnifOrderStat}}\label{sec:prooflem:boundUnifOrderStat}
We reproduce below the proof given in \cite{clemencon2025weak} for the sake of completeness. 

The order statistics of a
uniform random sample are sub-Gamma.  More precisely, it is shown in
\citet[Lemma 3.1.1.]{reiss2012approximate} that for every $k\le n$,
\begin{equation*}
  \label{eq:deviationOrderUni-reiss}
  \PP[\frac{\sqrt{n}}{\sigma }  \Bigl(1- U_{(k)} - \frac{k}{n+1}\Bigr)~\ge~t ]~\le~
  \exp\Big( -\,\frac{t^2}{3\bigl(1 + t/( \sigma\sqrt{n} ) \bigr)}\Bigr) ,
\end{equation*}
with $\sigma^2 = (1-k/(n+1))({k}/(n+1)) \le k/n$.  The above bound derives immediately from the cited reference and the fact that $1-U_{(k)} \eqd U_{(n+1-k)}$. 
Re-arranging we obtain  
\begin{align*}
  \PP[ 1 - U_{(k)} - k/n>t ] &~\le~ \PP[ 1 - U_{(k)} - k/(n+1)~>~t ] \\
  &~\le~\exp \Big( -\,\frac{nt^2/\sigma^2}{3(1 + t/ \sigma^2  )}\Big)\,.
\end{align*}
Inverting the above inequality  yields that with probability greater than $1-\delta$,
\begin{align*}
  1 - U_{(k)}&~\le~\frac{k}{n} + \sqrt{\frac{3 \sigma^2 \log(1/\delta)}{n} }+
               \frac{3\log(1/\delta)}{n} \\
          &~=~\frac{k}{n}\Big(1 + \sqrt{\frac{3 \log(1/\delta)}{k} }+
            \frac{3\log(1/\delta)}{k}  \Big).      \\
\end{align*}
\qed

\subsection{Proof of Theorem~\ref{thm:main}}\label{sec:proof_main_thm_hill}
In the setting of Lemma~\ref{lem:decomposHill}, combining
Lemmas~\ref{lem:decomposHill},
\ref{lem:boundUnifOrderStat},~\ref{lem:tailBound_abs_Zk_1} we obtain
the following upper bound on the error: with probability at least
$1-\delta$, from a union bound over two adverse events of probability $\delta/2$ each,
\begin{equation}
    \label{eq:upperBoundErrorIntermediate}
    \begin{aligned}
      |\HillEst(k) - \gamma|~\le~
      & \gamma   V_1(k,\delta/2) ~+~
       2 \bar a \biggl(
        \frac{n}{
        (k+1)\bigl(1+ R(k+1, \delta/2)\bigr)} \biggr) + \dots\\
      & \bigl(1+ \tilde V(k,\delta/2)\bigr)~
        \bar b\biggl( \frac{n}{(k+1)\bigl(1+ R(k+1, \delta/2)\bigr)}\biggr)\,,
      \end{aligned}
  \end{equation}

 %

  where  $\bar a,\bar b$ are as in Lemma~\ref{lem:decomposHill}, $R(k,\delta)$ is defined in Lemma~\ref{lem:boundUnifOrderStat}, $V_1(k,\delta)$ is the $(1 - \delta)$-quantile of $|Z_k-1|$ and  $\tilde V(k,\delta)\ge V_1(k,\delta)$ is defined in Lemma~\ref{lem:tailBound_abs_Zk_1}.
  The statement of the theorem now follows  immediately from~\eqref{eq:upperBoundErrorIntermediate} and the fact that the function $\tilde V$ (\emph{resp.} $R$) in that statement is a non increasing (\emph{resp.} non decreasing)  function of $k$.

\subsection{Proof of Lemma~\ref{lem:gamma_quantile_lower}}\label{sce:proof_gamma_quantile_lower}
The random variable $kZ_k$ has distribution
$\mathrm{Gamma}(k,1)$. Using~\citet[Theorem~5]{zhang2020non} gives
universal constants $0<c\le 1,\tilde C>0$ such that for any $x>0$,
  $$
\PP[kZ_k - k > x]~\ge~\max\bigl(ce^{- \tilde Cx}, ce^{-x^2/k} \bigr).  
$$
Thus 
  $$
\PP[|Z_k - 1|~>~x]~\ge~\PP[ kZ_k - k > k x ]~\ge~\max\bigl(ce^{- \tilde C k x}, ce^{-k x^2} \bigr).  
$$
Now, 
\begin{align*}
  V(k,\delta) &= \inf\{y: \PP[|Z_k-1 |\le y]~\ge~1-\delta/2 \}
              ~=~\inf\{y: \PP[|Z_k-1 |~>~y]~\le~\delta/2 \},   
\end{align*}
and the above argument yields  the following inclusion $$\{y\ge 0: \PP[|Z_k-1 |~>~y]~\le~ \delta/2 \}\subset
\{y~\ge~0: \max(ce^{-\tilde C y x}, ce^{-k y^2} )~\le~\delta/2\}.$$
This implies that
\begin{align*}
  V(k,\delta) &~\ge~\inf\bigl\{y\ge 0: \max\bigl(ce^{- \tilde C  k y }, ce^{-k y^2} \bigr) \le \delta/2\bigr\}   \\
              & = \inf \bigg\{~y~\ge~ 0~:~ y~\ge~ \frac{\log(2c/\delta )}{\tilde Ck} \quad \text{ and } \quad
                y~\ge~ \sqrt{\frac{\log(2c/\delta )}{k}} ~ \bigg\} \\
  &~=~\max\bigg(~\frac{\log(2c/\delta )}{\tilde Ck},~ \sqrt{0\vee\frac{\log(2c/\delta )}{k}}~  \bigg).  
\end{align*}
The statement follows with $c_2 = 2c$ and $c_1 = \min(1, 1/\tilde C)$. 

\qed
\subsection{Proof of Proposition~\ref{prop:lowerBoundKstar}}\label{sec:proof_prop_lowerBoundKstar}
Write $\grid = \{k_1< k_2 < \dotsb <  k_M\}$ with $M=|\grid|$ and let  $m^*$ denote the grid index such that $k_{m^*}=\kS(\delta,n)$.  Condition~\ref{cond:wideGrid} guarantees that  $m^* < M$. 
From the non-increasing property of $V$ and Lemma~\ref{lem:gamma_quantile_lower} we may bound the quantity $V(k,\delta)$ from below as follows, for any $1\le k\le n$,
\begin{align}
  V(k,\delta)~\ge~V(k+1, \delta)~\ge~  c_1 \Big(\sqrt{\frac{0\vee\log(c_2/\delta)}{k+1}} + \frac{\log(c_2/\delta)}{k+1}\Big)~\ge~c_1 \sqrt{0\vee\frac{\log(c_2/\delta)}{k+1}}. 
\end{align}

 Combining the latter  lower bound with the upper bound~\eqref{eq:bound-bias-vonMises} on the bias term we obtain the following implications,  
\begin{align}
  m ~>~ m^*  \iff&~B(k_m,n, \delta)  ~>~  \gamma   V(k_m,\delta)\nonumber \\
  ~\Rightarrow &~C_1(\delta,\rho) \Big(\frac{k_m+1}{n} \Big)^{-\rho} ~>~
                \gamma c_1 \sqrt{0\vee\frac{\log(c_2/\delta)}{k_m+1}} \nonumber \\
  \iff& ~ (k_m+1)^{-\rho + 1/2}  ~> ~   \frac{\gamma c_1\sqrt{0\vee\log(c_2/\delta)}}{C_1(\delta,\rho)} n^{-\rho}  \;
        \nonumber \\
  \iff& ~ (k_m+1)  ~> ~   \biggl(\frac{\gamma c_1\sqrt{0\vee\log(c_2/\delta)}}{C_1(\delta,\rho)} \biggr)^{\frac{2}{1-2\,\rho}}n^{-2\rho/(1-2\rho)}  \;\nonumber \\
  \Rightarrow& ~ k_m  ~\ge ~  
        \biggl(\frac{\gamma c_1\sqrt{0\vee\log(c_2/\delta)}}{C_1(\delta,\rho)} \biggr)^{\frac{2}{1-2\,\rho}} n^{-2\rho/(1-2\rho)} -1.
       \label{eq:condition_k_}  
\end{align}
Hence, because $k_{m^* + 1}$  satisfies~\eqref{eq:condition_k_},  we have 
\begin{align*}
   k_{m^* + 1}
    &~\ge~  \tilde C_2(\delta,\rho) ~ \gamma^{\frac{2}{1-2\,\rho}} ~
      n ^{\frac{ - 2 \rho }{1 - 2\,\rho}} - 1,
\end{align*}
where
$$\tilde C_2(\delta,\rho)~=~ \Big(c_1\sqrt{0\vee\log(c_2/\delta)} / C_1(\delta,\rho) ~\Big)^{\frac{2}{1 - 2\,\rho}}. $$
We obtain 
\begin{equation}
  \label{eq:explicitKS}
  \begin{aligned}
    \kS(\delta,n) &~=~k_{m^*}
    \ge \frac{k_{m^*}}{k_{m^*+1}} \left(\tilde C_2(\delta,\rho) ~ \gamma^{\frac{2}{1-2\,\rho}} ~
      n ^{\frac{ - 2 \rho }{1 - 2\,\rho}} -1\right)\,. 
  \end{aligned}
\end{equation}

We now derive an upper bound for $\tilde
C_2(\delta,\rho)$, leveraging the lower
bound~\eqref{eq:boundConstantDelta} on
$C_1(\delta,\rho)$.  Under 
Condition~\eqref{eq:condition_delta_for_lowerBound_C2} on
$(c_2,\delta)$ we get $c_2/\delta\ge \sqrt{4/\delta}$, whence
\begin{align}
    \tilde C_2(\delta,\rho)^{\frac{1-2\rho}{2}}
    &~\ge~  \frac{c_1\sqrt{\log(c_2/\delta)}}{
        C\Big( 1 + \sqrt{3\log(4/\delta)} +
                   3\log(4/\delta)\Big)^{1-\rho}}\nonumber \\
    &~\ge~ \underbrace{\frac{ c_1\sqrt{2^{-1}\log(4/\delta)}}{C\big(3\log(4/\delta)\big)^{1-\rho}}}_{C'_2} \times
     \underbrace{\frac{1}{ \big(1 + \frac{1}{\sqrt{3\log(4/\delta)} }+
                  \frac{1}{ 3\log(4/\delta)} \big)^{1-\rho}} }_{D} \;.\nonumber 
\end{align}
Simplifying $C'_2$ we get
\begin{align}
C'_2    &~=~ \frac{c_1}{C\sqrt{2}} \frac{\sqrt{ \log(4/\delta)}}{
    \big(3\log(4/\delta\big)^{1-\rho}} \nonumber \\ 
    &~\ge~\frac{c_1}{C3^{1-\rho}\sqrt{2}}\big(\log(4/\delta) \big)^{\frac{-(1- 2\rho)  }{2}}\;,
    \nonumber
\end{align}
and under~\eqref{eq:condition_delta_for_lowerBound_C2}, we have that $3\log(4/\delta)> 4$, so that 
\begin{align}
D    &~\ge~(4/7)^{1-\rho} . \nonumber 
\end{align}
Combining the two latter bounds we obtain 
$$
  \tilde C_2(\delta,\rho)^{\frac{1-2\rho}{2}} 
  ~\ge~\frac{c_1}{(21/4)^{1-\rho}\sqrt{2}C} \big(\log(4/\delta) \big)^{\frac{-(1- 2\rho)  }{2}}, 
$$
hence, 
$$
\tilde C_2(\delta,\rho)~\ge ~\Big(\frac{c_1^2}{2 C^2}\Big)^{1/(1- 2 \rho)}\times 
\frac{1}{(21/4)^2}\times\frac{1}{\log(4/\delta)}. 
$$
The result follows from the above display combined with~\eqref{eq:explicitKS} and the assumption in the statement that $ k_{m^*} / k_{m^*+1} \ge \beta^{-1}$. 


\subsection{Proof of Theorem~\ref{theo:errorOracle_vonMises}}\label{sec:prooferrorOracle_vonMises}
  From Proposition~\ref{prop:errorOracle}, for any $\delta>0$ and $n$ large enough so that 
  $\kS(\delta,n)\ge \log(4/\delta)$  
  it holds that    
\begin{align}
  \bigl|\hat \gamma~\bigl(\kS(\delta,n)\bigr) - \gamma\bigr|
  &~\le~2\gamma ~V\big(\kS(\delta,n), \delta \big) \nonumber \\
  &~\le~2 \gamma~ \Big(\sqrt{\frac{2\log(4/\delta)}{\kS(\delta,n)}} +
    \frac{\log(4/\delta)}{\kS(\delta,n)} \Big)\qquad \text{ (from~\eqref{eq:expressionV})}\nonumber \\
  &~\le~2 (1+\sqrt{2})\gamma~ \sqrt{\frac{\log(4/\delta)}{\kS(\delta,n)}}  \nonumber \\
  &~\le~ 2 (1+\sqrt{2})~\gamma ~\sqrt{\beta}~\sqrt{\frac{\log(4/\delta)}{\beta\,\kS(\delta,n) }}  \nonumber \\
  &~\le~2 (1+\sqrt{2})~\gamma~ \sqrt{\beta}~\sqrt{\frac{1+\log(4/\delta)}{\beta\,\kS(\delta,n) + 1}}  \nonumber \\
  &~\le~2 (1+\sqrt{2})\gamma ~\sqrt{\beta}~ \sqrt{1+\log(4/\delta)} \gamma^{\frac{-1}{1-2\,\rho}}C_2(\rho)^{-1/2} n ^{\frac{ \rho }{1- 2\, \rho}}  \quad\text{(from Proposition~\ref{prop:lowerBoundKstar})}\nonumber \\
  & ~=~  2 (1+\sqrt{2}) \;  ~\sqrt{\beta} \sqrt{1+\log(4/\delta)}\; C_2(\rho)^{-1/2}\; \gamma^{\frac{- 2\rho}{1-2\,\rho}} \; n ^{\frac{ \rho }{1- 2\, \rho}}. \nonumber
\end{align}


%% file: Contents2/Appendix.tex
In this appendix, we provide additional discussion about our simulations and complementary experiments.
\subsection{Perturbed Pareto distribution}
\label{sec:xt}
We provide some additional details regarding the perturbed Pareto distribution used in our experiments, see also \cite{resnick2007heavy}, 
defined by its cumulative distribution function $F(x)=1-cx^{-\alpha}(\log x)^{\beta}$, for  $x>x_0 = \exp(\alpha/\beta)$ with $\beta,\alpha>0$, with $c = (e \alpha/\beta )^{\beta}$. We show that $F$ thus defined is a valid distribution function. That $F(x)\to 1$ as $x\to\infty$ derives immediately from asymptotic comparison between powers of $x$ and of $\log x$. Now for $x\ge x_0$, the derivative of $F$ exists and is given by 
$f(x)=c(\log x)^{\beta-1}x^{-\alpha-1}(-\beta+\alpha\log x)\ge 0$. Finally, straightforward computations show that  $F(c)=0$. 

We now show that the quantile function $Q$ associated to $F$ admits the standardized Karamata representation~\ref{eq:karamataStrong}. 
We use the following characterization based on the so-called 
 `von Mises condition' as stated in standard reference books~\citep{de1970regular,resnick2007heavy,beirlant2006statistics}. A distribution function $F$ with density $f$ satisfies the von Mises condition if $\exists\ell>0$, such that
\begin{equation}\label{eq:vonMisesRatio}
\frac{x f(x)}{1-F(x)} \to \ell.
\end{equation}
From \citet[Theorem~2.7.1]{de1970regular},~\eqref{eq:vonMisesRatio} implies that $1-F$ is regularly varying with regular variation index $-\alpha = -\ell$, i.e. the tail index is $\gamma = 1/\alpha = 1/\ell$. 
Also from \citet[Lemma~6]{csorgo1985kernel} (and as noted in \cite{drees1998selecting}), more is true: \eqref{eq:vonMisesRatio} is equivalent to the condition   that the quantile function 
$Q(t) = F^\leftarrow(1-1/t)$  has Karamata representation \eqref{eq:karamataStrong}
which is nearly as strong as Condition~\ref{cond:vonMises}, although the additional requirement in  that $\bar b(t) \le C t^{\rho}$ is not guaranteed. 

 
In  the case of the `Perturb' distribution defined above, the von Mises ratio satisfies 
\begin{align*}
    \frac{x f(x)}{1-F(x)} 
    = \frac{(\log x)^{\beta-1}x^{-\alpha-1}(-\beta+\alpha\log x)}{x^{-\alpha}(\log x)^{\beta}} 
    &= (\log x)^{-1} (-\beta + \alpha \log x) \\
    & \xrightarrow[x\to\infty]{} \alpha, 
\end{align*}
which shows~\eqref{eq:vonMisesRatio}, thus ensuring that the Karamata representation~\eqref{eq:karamataStrong} holds.  

We now show that the additional condition that $\bar b(t) \le C t^{\rho}$ does not hold for any $\rho<0$n so that Condition~\ref{cond:vonMises} is not satisfied. We use the fact that \citep[see the discussion in ][following equation (5)]{drees1998selecting} in~\eqref{eq:karamataStrong}, one may choose 
$b(s) = q( Q(s)) - \gamma$ where $q(x) = \bar F(x)/(x f(x))$ is the inverse of the ratio in~\eqref{eq:vonMisesRatio}, and $Q(s) = F^\leftarrow(1 - 1/s)$. Now with the `perturb' distribution,
\begin{align*}
    q( Q(s)) - \gamma 
    &= \frac{\log Q(s)}{ -\beta + \alpha \log Q(s)} -1/\alpha\\
    & =\frac{\beta}{\alpha(-\beta + \alpha \log(Q(s))} \\
    &\ge \frac{\beta}{\alpha^2 \log Q(s)}.
\end{align*}
In addition,  because the function $Q$ is regularly varying with regular variation index $1/\alpha$ we have that for any $\epsilon>0$, for $s$ sufficiently large, 
$Q(s) \le c s^{1/\alpha+\epsilon}$ for some constant $c>0$ and we obtain 
\begin{align*}
    q( Q(s)) - \gamma
    &\ge \frac{\beta}{\alpha^2(\log  c + (1/\alpha + \epsilon) \log s)} 
    \\
    & \sim \frac{A}{\log s}
\end{align*}
where $A>0$. This well result is known as a Potter bound. The latter quantity converges to zero much slower than $s^\rho$ for any $\rho<0$. We conclude that for the perturb distribution, we cannot have $\bar b(t) < C t^\rho$ for some $\rho<0$, thus Condition~\ref{cond:vonMises}
 is not satisfied.

\subsection{Counter-example distribution}\label{sec:counterExample}
The general idea behind the proposed counter-example distribution is
to start from a well behaved survival function $1-F$ and to modify it
and `draw holes' in the support, by pushing back some mass from some
intervals to neighboring intervals.
This strategy is suggested by well-known facts: From
\citet[Theorem~2.7.1 (b)]{de1970regular}, if $1-F$ is regularly
varying and if in addition, $F$ has a density which is ultimately
non-increasing, then~\eqref{eq:vonMisesRatio} holds true, thus, by the
argument developed in Section~\ref{sec:xt}, also the standardized
Karamata representation~\eqref{eq:karamataStrong} is satisfied. As a
consequence, to construct a distribution function $F$ such that $(i)$
$1-F$ is regularly varying; but $(ii)$ the standardized Karamata
representation does not hold, and if we would like to stay in the
continuous case where a density exists, so as to avoid additional
difficulties with potential ties in the order statistics, then the
density $f$ must \textit{not} be ultimately non-increasing.

Consider the random variable $X$ defined by 
\begin{equation}
    \label{eq:counterExampleGenerativeDefinition}
    X=  \lfloor Z^s  \rfloor^{1/s  } + 
  \frac{1}{2}(Z -  \lfloor Z^s   \rfloor^{1/s  }), 
\end{equation}
where $Z\sim\textrm{Pareto}(1/\gamma)$ for some $\gamma>0$, \ie $\PP[Z>t] = t^{-1/\gamma}, t\ge 1$ and $s \in(0,1]$ is a scaling  parameter.
 For $z\ge 1$, define  $$r(z) = z- \lfloor z^s  \rfloor^{1/s  }.$$ Note that $r(z)\ge 0$ and that the discrepancy between
$Z$ and $X$ writes
$
Z-X = \frac{1}{2} r(Z).
$ 
For $z \in [n^{1/s}, (n+1)^{1/s})$, we have 
$ 
0\le r(z) \le (n+1)^{1/s} - n^{1/s} 
$ 
and straightforward calculations show that 
$$
\liminf_{z\to\infty} \frac{r(z)}{s^{-1}z^{1-s}} = 0 \quad;\quad
\limsup_{z\to\infty} \frac{r(z)}{s^{-1}z^{1-s}} = 1. 
$$
Thus for $s<1$ the discrepancy between $X$ and $Z$ becomes more and more pronounced as $Z$ becomes large. 

Let us denote by $H$ the cumulative distribution function of $X$ in the remainder of the proof. Consider the intervals 
$$I_n = \Big[n^{1/s},  \frac{n^{1/s} + (n+1)^{1/s}}{2} \Big] ~~;~~ 
J_n  = \Big(\frac{n^{1/s} + (n+1)^{1/s}}{2}, (n+1)^{1/s} \Big).$$
Then the positive real line is the disjoint union of the $I_n,J_n$'s and 
we have 
$$
X ~=~\begin{cases}
n^{1/s} + \frac{1}{2}r(Z) &  \text{ if } Z\in I_n\\
 \frac{n^{1/s} + (n+1)^{1/s}}{2} & \text{ if } Z\in J_n.
\end{cases}
$$
so that $\PP[X \in \bigcup_{n\ge 0} J_n] = 0$. As a consequence $X$
admits a density that vanishes on the intervals $J_n$ while it is
non-zero on the $I_n$'s. 
Writing $F(t) = \PP[Z\le t] = 1 - t^{-1/\gamma}$, the identity in the above display shows that 
\begin{equation*}
    H(x)~=~\begin{cases}
        F\bigl(\lfloor x^s   \rfloor^{1/s  } + 2r(x)\bigr) & \text{if }
        r(x)\le \frac{(\lfloor x^s\rfloor + 1)^{1/s} - \lfloor x^s\rfloor^{1/s}}{2}, \\
        F\bigl(\frac{\lfloor x^s\rfloor^{1/s} + (\lfloor x^s\rfloor + 1)^{1/s}}{2} \bigr)& \text{if } r(x)> 1/2. 
    \end{cases}
\end{equation*}
Because $r$ has derivative $r'(x) = 1$ on $I_n$, we deduce that $H$ has density
$$
h(x) = 2 f\bigl(\lfloor x^s \rfloor^{1/s} + 2r(x)\bigr) = 2 \alpha
\bigl(\lfloor x^s\rfloor^{1/s} + 2r(x)\bigr)^{-\alpha-1},$$ here
$f(x) = \alpha x^{-\alpha - 1}$.  Thus on such an interval $I_n$, we
have
$$
x h(x)/\bigl(1-H(x)\bigr
)~=~ \frac{2x\alpha  \bigl(\lfloor x^s\rfloor^{1/s} +2r(x)\bigr)^{-\alpha-1}}{
  \bigl( \lfloor x^s\rfloor^{1/s} +2r(x)\bigr)^{-\alpha}
} ~\sim_{n\to\infty}~ 
\frac{2 \alpha x x^{-\alpha-1}}{x^{- \alpha }}\to 2\alpha. 
$$
Thus the von Mises ratio in the left-hand side
of~\eqref{eq:vonMisesRatio} jumps from $0$ on $J_n$ to a level close
to $2\alpha$ infinitely often, and it cannot have a limit, which proves that~\eqref{eq:karamataStrong} cannot hold for the distribution of $X$.

We now check that $1-H$ is regularly varying with regular variation index $-\alpha$.  Since $H$ and $F$ coincide
on the set $\{n^{1/s}, n\in\nset\}$, we have for $x>0,t>0$:
\begin{align*}
  \frac{ (\lfloor (tx)^s\rfloor +1)^{-\alpha/s} }{
  \lfloor t^s\rfloor  \bigr)^{-\alpha/s} }
  ~=~
  \frac{1 - F\big((\lfloor (tx)^s\rfloor  +1)^{1/s}\big)}{1 -F(\lfloor t^s\rfloor^{1/s} )}
  \le \frac{1 - H(tx)}{1 -H(t)}~\le~ 
  \frac{1 - F(\lfloor (tx)^s \rfloor^{1/s}  )}{
  1 -F\big( (\lfloor t^s\rfloor + 1)^{1/s}\big) } ~=~ 
  \frac{\lfloor (tx)^s\rfloor^{-\alpha/s} }{(\lfloor t^s\rfloor + 1)^{-\alpha/s}}
\end{align*}
and both sides of the sandwich converge to $x^{-\alpha}$ as $t\to\infty$, showing that 
$(1-H(tx))/(1-H(t))\to x^{-\alpha}$ as $t\to \infty$, which concludes the proof.

%% file: Contents/Addsim.tex
\subsection{Additional Simulations}\label{sec:adsim}

We repeat the experiments from Section~\ref{sec:sim}, this time using a linear grid defined by $k_m = \lfloor m n/|\grid| \rfloor$ for $1\le m\le |\grid|$, where $|\grid|=\log n/\log(\beta)$ and $\beta=1.1$. The results, shown in Table~\ref{tab:Simvonmises30L} and Table~\ref{tab:Simvonmises30QL}, indicate that our EAV estimator is largely robust to changes in the grid.
\begin{table}[H]
	\centering
	\scalebox{1}{
		\begin{tabular}{ccccc}
			\hline
			d.f. & $\gammaS$    & $\kAV;\bar{k}_{\textrm{EAV}}$  &  $\kbch;\bar{k}_{\textrm{BT}}$ &   $\hat{k}_{DK};\bar{k}_{\textrm{DK}}$  \\
    \hline
            $C_{2,2/3}$& 0.5&  [0717,1051];0922 & [31,0127];0066 & [001,0011];0007\\
            \hline
             $C_{2,1/2}$& 0.5&  [0368,1399];0595 & [31,0109];0043 & [001,0241];0046\\
            \hline
   $S_{1.7}$ & $1/1.7$ &  [0652,3991];3427  & [35,2253];0515 &[323,0729];0495\\
			\hline
			$S_{1.5}$ & $1/1.5$ &  [0717,4830];4356  &  [33,3221];1154 & [534,1298];0935\\
              \hline
   $S_{1.99}$ & $1/1.99$ &  [0072,1862];1371  & [42,0679];0281 &[002,0102];0035\\
      \hline
			$L_{2,1}$  & 0.5 &  [4830,7778];6605 &  [35,3854];1230 &   [250,1974];1132 \\
			\hline 
            			\hline
			$F_{1,10}$ & $1.0$ &  [0717,1862];1334   & [33,1064];0459 &  [042,0121];0068\\
			\hline
   			
			$F_1$& $1.0$&  [5313,7071];6216  &  [38,4195];1430 &  [853,2837];1665\\
			\hline
            $PCP_{1.1}$& $1.1$ &[0717,8500];4663  &   [50,9900];3727 & [3908,4605];4340\\
            \hline
            $PCP_{1.5}$& $1.5$ &[0789,3991];2200  &   [47,1057];0694 & [0073,5083];2890\\
            \hline
            $PCP_{1.25}$& $1.25$ &[0717,9999];7111  &   [33,9998];0723 & [4048,5050];4519\\
            \hline
	\end{tabular}}
	\caption{Range of values of $k$ output by $\kAV,\kbch,\kdk$, and mean values $\bar{k}_{\textrm{EAV}},\bar{k}_{\textrm{BT}},\bar{k}_{\textrm{DK}}$, over $N=500$ experiments 
   }
	\label{tab:Simvonmises30Q}
	
\end{table}
\begin{table}[H]
	\centering
	\scalebox{0.8}{
		\begin{tabular}{cc|ccc}
			\hline
&  &&$n=10\,000$ & \\
			d.f. & $\gammaS$ &  $\widehat{\textrm{MSE}}(\kAV);(\operatorname{stderr}) \times 100$  &$\widehat{\textrm{MSE}}(\kbch);
            (\operatorname{stderr}) \times 100$&  $\widehat{\textrm{MSE}}(\hat{k}_{DK});
            (\operatorname{stderr}) \times 100$\\
   \hline
            $C_{2,2/3}$& 0.5     & $\textbf{01.51};(0.04)$   & $03.05;(0.16)$  & $13.80;(0.88)$ \\
            \hline
            $C_{2,1/2}$& 0.5     & $\textbf{04.92};(0.09)$   & $05.14;(0.24)$  & $29.01;(5.99)$ \\
			\hline
			$S_{1.7}$ & $1/1.7$  & $\textbf{00.61};(0.03)$  & $04.01;(0.10)$ & $04.40;(0.06)$ \\
			\hline
			$S_{1.5}$ & $1/1.5$ &$\textbf{00.59};(0.01)$  & $01.33; (0.05)$  & $01.23;(0.03)$\\
			\hline
            $S_{1.99}$ & $1/1.99$ &$\textbf{21.79};(0.09)$  & $47.70;(0.24)$  & $45.88;(0.48)$\\
			\hline
			$L_{2,1}$  & $0.5$  & $22.79;(0.15)$  & $04.47;(0.14)$ & $\textbf{04.40};(0.08)$\\
			\hline
            \hline 
			$F_{1,10}$ & $1.0$  & $13.65;(0.12)$   & $03.49;(0.09)$  & $\textbf{01.78};(0.12)$ \\
			\hline
			$F_1$& $1.0$& $07.11;(0.07)$  &  $00.69;(0.05)$ &  $\textbf{00.24};(0.03)$  \\
            \hline 
            $PCP_{1.1}$ & $1.1$ & $\textbf{00.63};(0.01)$  & $00.70;(0.11)$ & $00.73;(0.04)$ \\
			\hline
            $PCP_{1.5}$ & $1.5$ & $05.57;(0.05)$  & $\textbf{00.58};(0.07)$ & $05.99;(0.36)$ \\
			\hline
            $PCP_{1.25}$ & $1.25$ & $03.51;(0.02)$  & $\textbf{01.13};(0.06)$ & $03.34;(0.02)$ \\
			\hline
	\end{tabular}}
	\caption{
    $\widehat{\textrm{MSE}}(\hat k)$ (stderr $(\hat k)$), both multiplied by $100$, for $\hat k \in\{\kAV,\kbch,\kdk\}$, over $N=500$ experiments. 
    }  
	\label{tab:Simvonmises30L}	
\end{table}

\begin{table}[H]
	\centering
	\scalebox{1}{
		\begin{tabular}{ccccc}
			\hline
			d.f. & $\gammaS$    & $\kAV;\bar{k}_{\textrm{EAV}}$  &  $\kbch;\bar{k}_{\textrm{BT}}$ &   $\hat{k}_{DK};\bar{k}_{\textrm{DK}}$  \\
     \hline
            $C_{2,2/3}$& 0.5&  [0729,1145];0967 & [31,0127];0066 & [001,0011];0007\\
            \hline
             $C_{2,1/2}$& 0.5&  [0312,1562];1118 & [31,0109];0043 & [001,0241];0046\\
            \hline
   $S_{1.7}$ & $1/1.7$ &  [0937,4375];3724  & [35,2253];0515 &[323,0729];0495\\
			\hline
			$S_{1.5}$ & $1/1.5$ &  [1354,5000];4591  &  [33,3221];1154 & [534,1298];0935\\
            \hline
			$S_{1.99}$ & $1/1.99$ &  [1250,2083];1653  &  [42,0679];0281 & [002,0102];0035\\
      \hline
			$L_{2,1}$  & 0.5 &  [5312,8437];6970 &  [35,3854];1230 &   [250,1974];1132 \\
			\hline 
            			\hline
			$F_{1,10}$ & $1.0$ &  [0833,1979];1447   & [33,1064];0459 &  [042,0121];0068\\
			\hline
   			
			$F_1$& $1.0$&  [5208,7395];6381  &  [38,4195];1430 &  [853,2837];1665\\
			\hline
            $PCP_{1.1}$& $1.1$ &[1770,9687];6429  &   [50,9900];3727 & [3908,4605];4340\\
            \hline
            $PCP_{1.5}$& $1.5$ &[1354,4166];2362  &   [47,1057];0694 & [0073,5083];2890\\
            \hline
             $PCP_{1.25}$& $1.25$ &[1562,9998];7831  &   [33,9998];0723 & [4048,5050];4519\\
            \hline
	\end{tabular}}
	\caption{Range of values of $k$ output by $\kAV,\kbch,\kdk$, and mean values $\bar{k}_{\textrm{EAV}},\bar{k}_{\textrm{BT}},\bar{k}_{\textrm{DK}}$, over $N=500$ experiments}
	\label{tab:Simvonmises30QL}
	
\end{table}